\documentclass[final,11pt]{amsart}
\usepackage[T1]{fontenc}
\usepackage{a4wide,mathtools,bbm,amsthm,amsmath,amsfonts}
\usepackage{amssymb,color,epic,multirow,comment}
\usepackage{bm}
\usepackage[noend]{algpseudocode}
\usepackage[ruled]{algorithm2e}
\usepackage{graphicx}
\usepackage{subcaption}
\usepackage[export]{adjustbox}
\usepackage{tikz,pgfplots}
\usepackage{orcidlink}
\usepackage{float}
\usepackage{booktabs}
\usepackage{tikz-cd, comment}
\usepgfplotslibrary{groupplots}
\pgfplotsset{compat=newest}
\usetikzlibrary{arrows,decorations,backgrounds,positioning,arrows.meta,
decorations.markings}
\usepgfplotslibrary{colorbrewer}

\mathtoolsset{showonlyrefs}
\theoremstyle{plain}

\newtheorem{theorem}{Theorem}[section]
\newtheorem{lemma}[theorem]{Lemma}

\newtheorem{definition}[theorem]{Definition}

\def\Letters{A,B,C,D,E,F,G,H,I,J,K,L,M,N,O,P,Q,R,S,T,U,V,W,X,Y,Z}
\makeatletter
\@for \@l:=\Letters \do{%
  \expandafter\edef\csname\@l bb\endcsname{%
  \noexpand\ensuremath{\noexpand\mathbb{\@l}}}%
  \expandafter\edef\csname\@l bf\endcsname{{\noexpand\bf \@l}}%
  \expandafter\edef\csname\@l cal\endcsname{%
  \noexpand\ensuremath{\noexpand\mathcal{\@l}}}%
  \expandafter\edef\csname\@l eu\endcsname{%
  \noexpand\ensuremath{\noexpand\EuScript{\@l}}}%
  \expandafter\edef\csname\@l frak\endcsname{%
  \noexpand\ensuremath{\noexpand\mathfrak{\@l}}}%
  \expandafter\edef\csname\@l rm\endcsname{{\noexpand\rm \@l}}%
  \expandafter\edef\csname\@l scr\endcsname{%
  \noexpand\ensuremath{\noexpand\mathscr{\@l}}}%
}
\newcommand{\bs}[1]{{\boldsymbol#1}}

\newcommand{\isdef}{\mathrel{\mathrel{\mathop:}=}}

\newcommand{\bbN}{{\mathbb{N}}}
\newcommand{\bbP}{{\mathbb{P}}}
\newcommand{\bbR}{{\mathbb{R}}}

\newcommand{\calC}{{\mathcal{C}}}
\newcommand{\calE}{{\mathcal{E}}}
\newcommand{\calL}{{\mathcal{L}}}
\newcommand{\calT}{{\mathcal{T}}}

\definecolor{navy}{RGB}{102,153,255}
\definecolor{tuerkis}{RGB}{51,153,204}
\algrenewcommand\alglinenumber[1]{\ding{\numexpr191 + #1}}

\newcommand{\node}[1]{\mathrm{#1}}
\title[
$hp\,$-adaptive trees for graph signal
approximation]{$\bs h\bs p\textnormal{\,}$-adaptive trees for graph signal
approximation}
\author{Giacomo Elefante${\hphantom{}}^1$ \orcidlink{0000-0001-5576-6802}, 
Wolfgang Erb${\hphantom{}}^2$ \orcidlink{0000-0003-3541-5401}, 
Michael Multerer${\hphantom{}}^1$ \orcidlink{0000-0003-0170-0239}}
\address{
${\hphantom{}}^1$ IDSIA USI-SUPSI, Universit{\`a} della Svizzera italiana,
Lugano, Switzerland }
\address{
${\hphantom{}}^2$ Department of Mathematics “Tullio Levi-Civita” and Padova
Neuroscience Center (PNC), Universit{\`a} di Padova, Padova, Italy}
\email{giacomo.elefante@usi.ch, wolfgang.erb@unipd.it, michael.multerer@usi.ch}
\begin{document}
\begin{abstract}
Tree-encoded partitionings of graphs are fundamental tools for the decomposition
and approximation of graph signals.
For the efficient approximation of such graph signals, we develop strategies
based on $hp$-refinement by combining domain decomposition with an improved
local approximation using polynomials of higher degree. In this way,
from a given
graph partitioning tree, a more efficient subtree is extracted in which the
cost of the signal approximation is considerably reduced by still
maintaining the same total error. To this end, we interpret the refinement
process as a binary knapsack problem to determine an enhanced partitioning tree.
We further study an a-posteriori strategy which prunes the partitioning 
tree by optimizing the polynomial degrees over the subdomains. To make
polynomial basis systems accessible for general graphs or high-dimensional data,
we propose local embeddings of graphs into low dimensional Euclidean spaces.
We underpin the efficiency of our algorithms with extensive numerical tests
which carefully assess the impact of the applied refinements and optimization
strategies. 
\end{abstract}

\maketitle

\section{Introduction}
\subsection{Motivation}
Modern sensing and imaging technologies produce large amounts of data on
general domains that are naturally modeled as signals on graphs
\cite{O2022,OFKMV18,SNFOV13,SDS2019}. Examples include seismic observations,
point clouds acquired by Lidars or 3D scanners, as well as transportation and
power networks. Efficient storage, processing, and analysis of these
data require approximation algorithms that adapt to the geometry of
the graph as well as to the local regularity of the signal values.

Efficient algorithms for the approximation of these signals use
multiscale approaches and wavelet-type constructions based 
on a hierarchy of nested partitions of the graph domain 
\cite{Erb23,GNC10,HM22,REC11}. These approaches rely
on the refinement of the partitions to improve resolution and, consequently,
approximation. However,
approximation theory for finite elements has shown that simultaneously adapting
the partition size ($h$-adaptivity) and the local polynomial degree
($p$-adaptivity) often yields significantly sparser representations,
particularly for functions with spatially varying regularity. Generalizing this
$hp$-adaptive approach to graph signal processing is therefore a natural step
toward a more efficient compression and approximation of graph signals.

\subsection{Literature review}
Partition-based approaches recursively decompose the graph domain into smaller
subdomains, so that a partitioning tree is constructed that contains the
subdomains of the partitions as nodes. Graph signals are then approximated by a
function from a simple local approximation space, e.g., piecewise constants,
on each element of the
partition. This hierarchical
decomposition induces a natural multiresolution analysis that forms the basis
for spatially localized wavelets and related basis systems on point clouds and
graphs \cite{GNC10,HM22,IS2014,NO2012,REC11}. The main advantage of this
approach compared to other wavelet constructions such as diffusion wavelets
\cite{CM06} and spectral graph wavelets \cite{HVG11} is its computational
efficiency and the clear spatial localization of the basis elements in terms of
the given partitions.

Adaptive partition-based methods also incorporate information from the signal.
Here, the partitioning tree is refined according to the local properties of the
signal, generating finer partitions only where the local error is large.
This strategy resembles binary space partitioning methods developed for image
processing \cite{RLVN1991}, or the wedgelet construction introduced in
\cite{Donoho99,FDFW07}, where adaptive geometric partitions yield sparser
descriptions of signals. These adaptive strategies have been adapted to graph
signals in \cite{Erb23} in terms of so-called graph wedgelets.

All these constructions are pure $h$-refinement methods, which decrease the
partition size $h$ to increase the resolution of the approximation.
This restriction to domain decompositions leads typically to algebraic
convergence rates governed by the smoothness class of the signal
\cite{BinevDeVore04,DeVore98,Erb23}. In many practical scenarios, signals
however display a rather heterogeneous behavior, incorporating regions with
large smoothness and other components with oscillatory or singular behavior.
In finite element methods, this limitation is overcome by mixing $h$-refinement
with a so-called $p$-refinement, where the approximation is improved by
adaptively increasing the degree of the local polynomial approximant. The
resulting $hp$-refinement is classical and extensively studied in the finite
elements literature, see, e.g.,
\cite{BabuvskaSuri94,MitchellMcClain14,book:Schwab98,TS03}.
While mesh refinement is robust near singularities, high-order approximation can
deliver up to exponential convergence in regions with a smooth signal
\cite{ApelMelenk17}.

Since the refinement tree of an adaptive mesh encodes
all subdivisions of the domain, 
tree-based approximation frameworks for $hp$-adaptivity and near-best, cost-aware
refinement can naturally be accommodated \cite{Binev18,BDL11}.
Profit-to-cost tree selections, formulated as binary knapsack problems
\cite{book:MartelloToth90}, have been effectively applied for quasi-optimal
polynomial and sparse-grid approximation of partial differential equations with
random data \cite{BTNT12,NTT16}. Further, bottom-up error-cost pruning methods are
well-established across various fields. Notable examples are the classical
minimal cost-complexity pruning decision trees \cite[Section 3.3]{book:Breiman},
as well as the threshold-based pruning of affine Walsh transforms, as for
instance described in \cite{Florez23}. Finally, local low-dimensional
embeddings of graphs can be obtained via classical nonlinear dimension reduction
techniques such as multidimensional scaling \cite{book:coxcox, book:LeeVerleysen} and
\texttt{ISOMAP} \cite{ST02,TdSL00}, as well as with neighbor embedding
techniques such as \texttt{t-SNE} \cite{vdMH08} and \texttt{UMAP} \cite{MHM20}.

\subsection{Contributions}
Starting from the $hp$-refinement idea, we develop $hp$-adaptive tree
constructions for the approximation of graph signals, using graph wedgelets as
the primary tool to generate an initial partitioning tree. This tree is
then enhanced by two complementary strategies incorporating a $p$-adaptive
component based on polynomial refinements of higher degree. Since 
polynomial basis functions are not directly available on general graphs, we will 
use local embeddings in the Euclidean space to generate the local polynomial 
approximation spaces for the $p$-refinement.
Our contributions are the following.
\begin{itemize}
\item We establish that greedy $h$-refinement of the leaf
with the maximal local error is near-best in the class of $h$-adaptive trees
of comparable size under suitable assumptions on the local error, resembling a
well-known result from \cite{Binev18}.
\item We propose a greedy subtree selection based on a profit-cost trade-off,
which we formulate as a binary knapsack problem.
\item We propose an a-posteriori error-cost pruning (ECP) step, in which the
already refined tree is pruned bottom up and an $h$-approximation is
replaced by a $p$-approximation whenever this proves to be more efficient.
\item We show that the pruning phase in the ECP approach increases neither the
error nor the storage cost for the coefficients in an expanded tree. 
\item We consider patch-wise local embeddings to construct
the polynomial spaces on general graphs.
If low-dimensional local embeddings exist, this strategy
mitigates the curse of dimensionality and is therefore
particularly attractive for high-dimensional signals.
\end{itemize}
\subsection{Outline}
The remainder of this article is structured as follows. In Section
\ref{sec:wedgelets}, we introduce binary graph partitioning (BGP) trees to
encode hierarchical partitioning structures on a graph and recapitulate graph
wedgelets, whose adaptive BGP trees are used to validate the introduced
$hp$-strategies. In Section \ref{sec:adapt_split}, we show that greedy
$h$-refinement of the leaf with the largest error produces a near-optimal
$h$-adaptive tree. In Section \ref{sec:hp_strategies}, we introduce the greedy
knapsack criterion and the error-cost pruning method. In Section
\ref{sec:gen_graph}, we extend the construction to general graphs and
high-dimensional signals through a patch-wise embedding and a forest of graph
partitioning trees. In Section \ref{sec:num_test}, we validate all introduced
methodologies with extensive numerical tests on real-world and synthetic
datasets and compare them with a further $hp$-tree construction available in
the literature. Finally, we conclude the work in Section \ref{sec:conclude}.

\section{Binary Graph Partitioning Trees and Wedgelet Approximation}
\label{sec:wedgelets}

We encode hierarchical partitions of a vertex set utilized in
partition-based approximation methods by means of binary graph partitioning
(BGP) trees. This framework can naturally be extended to non-binary trees, but
for simplicity we restrict our attention to the binary case. Throughout this
work, we assume that $G=(V,E)$ is a simple undirected graph with vertex set
$V$ of cardinality $n$, edge set $E$, and graph distance $\mathrm{d}$, which
measures the distance between pairs of vertices.  

\begin{definition} \label{def:BGP} 
A binary graph partitioning (BGP) tree $\mathcal{T}$ of a graph $G$ with $L$
leaves is a binary tree consisting of subsets of the vertex set $V$ that are
ordered recursively in partitions $\mathcal{P}^{(m)}$, $m \leq L$, of $V$ by
the following rules:
\begin{enumerate}
\item The root $\Delta_1^{(1)}$ of the BGP tree $\mathcal{T}$ consists of the
  vertex set $V$. The respective first partition of $V$ is given as 
$$\mathcal{P}^{(1)} \isdef \{ \Delta_1^{(1)}\}.$$
\item For a partition $\mathcal{P}^{(m)} 
  = \{\Delta_1^{(m)}, \ldots, \Delta_m^{(m)}\}$ of $V$ consisting of $m < L$
elements in the BGP tree $\mathcal{T}$, the next partition $\mathcal{P}^{(m+1)}$
of $V$ in $\mathcal{T}$ is obtained by a dyadic split of a chosen set
$$\Delta_j^{(m)} = \Delta_{j}^{(m+1)} \cup \Delta_{m+1}^{(m+1)}, 
\quad j \in \{1, \ldots, m\},$$ into two disjoint non-empty sets 
$\Delta_{j}^{(m+1)}$, $\Delta_{m+1}^{(m+1)}$. The next partition 
$\mathcal{P}^{(m+1)}$ is given as  
\[
\mathcal{P}^{(m+1)} \isdef \{\Delta_1^{(m+1)}, \ldots, \Delta_{m+1}^{(m+1)}\},
\]
with the $m-1$ remaining sets given as $\Delta_{k}^{(m+1)} = \Delta_{k}^{(m)}$ 
for $k \in \{1, \ldots, m\} \setminus \{j\}$.
\end{enumerate}  
\end{definition}

If $\Delta_0$ and $\Delta_1$ are elements of $\mathcal{T}$ obtained from a
dyadic split of an element $\Delta \in \mathcal{T}$, then $\Delta_0$ and
$\Delta_1$ are referred to as children of $\Delta$ in the tree $\mathcal{T}$. A
BGP tree $\mathcal{T}$ is by construction a full binary tree in the sense that a
set $\Delta \in \mathcal{T}$ can only have two children or no children at all.
In the latter case, we call $\Delta$ a leaf of the tree $\mathcal{T}$. The set
$\mathcal{L}(\mathcal{T})$ of all leaves of the tree $\mathcal{T}$ corresponds
to the final partition $\mathcal{P}^{(L)}$ in the splitting process and contains
$L$ elements.

\begin{definition}
A BGP tree $\mathcal{T}$ is called balanced if there exists a constant
$\frac12 \leq \rho < 1$ such that for every child $\Delta_0$ of
$\Delta \in \mathcal{T}$ we have
\[(1-\rho) |\Delta| \leq |\Delta_0| \leq \rho |\Delta|.\]
We call a BGP tree $\mathcal{T}$ complete if $\mathcal{L}(\mathcal{T})$ contains
exactly all $n$ vertices of $G$ as leaves. 
\end{definition}

In the graph wedgelet construction introduced in \cite{Erb23}, the BGP tree is
generated adaptively based on the information of the underlying graph signal.
At each refinement step $m$, the partition element $\Delta_j^{(m)}$ is selected
having the smallest local ${L}^2$-error for the signal relative to its
mean value. The chosen set $\Delta_j^{(m)}$ is then subdivided by a wedge split
chosen to minimize the sum of the squared approximation errors between the
signal and its mean values on the two resulting wedges. These wedge splits and
the corresponding partitions can be encoded compactly in terms of a sequence of
center nodes of the graph.

Once the adaptive BGP tree is generated by a hierarchic concatenation of wedge
splits, the wedgelet approximations of the signal consist of a piecewise
constant functions on the partitions $\mathcal{P}^{(m)}$ of the BGP tree.
Furthermore, the detail information between two approximations on consecutive
partitions $\mathcal{P}^{(m)}$ and $\mathcal{P}^{(m-1)}$ can be represented
compactly in terms of geometric wavelets. For the best $m$-term approximation
$S_m(f)$ of a signal $f$ in terms of $m$ geometric wavelet, it is shown in
\cite{Erb23} that the ${L}^2$-error decays algebraically as
\begin{equation} \label{eq:mtermalgebraicrate}
\left\| f - S_m(f) \right\|_{{L}^2} = \mathcal{O}(m^{-\alpha})   
\end{equation} 
for some $\alpha > 0$ depending on the Besov-type smoothness class of the signal
$f$. Such algebraic decay rates are well-known for adaptive $m$-term
approximation and can be found in classical contexts of non-linear approximation
\cite{DeVore98}, but also in continuous scenarios related to piecewise polynomial
and wavelet approximation based on hierarchical partitions \cite{DL2005,KP2003}.

\section{Greedy $h$-refinements of partitioning trees} \label{sec:adapt_split}
In the subsequent sections, we assume that $\mathcal{T}$ is a fixed BGP tree with $L$ leaves as given in Definition \ref{def:BGP}, obtained by recursively splitting a vertex set $V$ in a nested sequence of partitions. As a guiding example, one may consider a BGP tree constructed adaptively by wedge splits, as described in the previous section. The $h$- and $hp$-refinement strategies considered in this and next section will lead to different subtrees extracted from the initial tree $\mathcal{T}$. 

\subsection{Near-best $h$-approximation}
In the greedy construction of the BGP wedgelet tree $\mathcal{T}$, the largest local approximation error is used to select the partition element $\Delta_j^{(m)}$ to be split in the refinement step $m$. This selection rule provides already a first intrinsic greedy $h$-refinement procedure for the given BGP tree $\mathcal{T}$.

In the following, we analyse the approximation properties of this particular $h$-refinement procedure. 
In particular, we show that this greedy refinement strategy automatically produces subtrees 
that are near-optimal among all trees with comparable number of leaves under two assumptions on the local approximation errors. 

We associate to
each node $\Delta \in \mathcal{T}$ a local error $e(\Delta)$ and measure the global
approximation error linked to the finite tree $\calT$ by summing the local contributions among
its leaves in $\calL(\calT)$. In particular, we define the total error as
$$ 
\calE(\calT) = \sum_{\Delta\in\calL(\calT)} e(\Delta).
$$
The first assumption on the error functional is natural and concerns the superadditivity of the local 
error, i.e.,
\begin{equation} \label{eq:errorassumption1}
e(\Delta) \geq e(\Delta_0) + e(\Delta_1),
\end{equation}
where $\Delta_0$ and $\Delta_1$ denote the two children of $\Delta$. 
The second assumption formalizes the requirement of an even distribution of the error among the 
children of a node, this means that there exists a constant $c \in (0, 1/2]$ such that
\begin{equation} \label{eq:errorassumption2}
    e(\Delta_i) \ge c\, e(\Delta) \qquad\text{ for } i\in \{0,1\},
\end{equation}
meaning that the error cannot be completely localized in a single child.

The next two lemmas provide the upper and lower estimates on sums of local
errors that will be used in the proof of the main result below.

\begin{lemma} \label{Lemma1}
Let $t\geq 0$, and let $\Tcal$ be a BGP tree such that 
$e(\Delta)\leq t$ for all leaves $\Delta\in\calL(\Tcal)$. Then,
$$ \sum_{\Delta\in\calL(\calT)} e(\Delta) \leq L \, t 
= \frac{\#\Tcal +1}{2} \, t, $$
where \(\#\Tcal\) denotes the number of nodes in \(\Tcal\) and $L = \#\calL(\Tcal)$ is the number of leaves.
\end{lemma}
\begin{proof}
The sum on the left consists of $L = \#\calL(\Tcal)$
terms, each summand being bounded from above by $t$. This yields already the first 
inequality. Further, the number of nodes $\#\Tcal$ in a full binary tree is
given by $\#\Tcal = 2 L - 1$. The latter identity can be proven 
inductively: pruning $2$ sibling leaves from a full binary tree gives a new full binary 
tree with $L - 1$ leaves. 
\end{proof}

\begin{lemma} \label{Lemma2}
Let $a\geq 0$ and let $\widetilde{\Delta}$ be a node in the BGP tree $\calT$ 
and let $T$ be a full binary subtree of $\calT$ rooted at $\widetilde{\Delta}$ such that 
$e(\Delta)\geq a$ and the superadditivity \eqref{eq:errorassumption1} is satisfied for all nodes $\Delta \in T$. Then, we have
\[
e(\widetilde{\Delta})\geq \#\calL(T) a.
\]
\end{lemma}

\begin{proof}
The proof follows directly from the superadditivity \eqref{eq:errorassumption1} of the error functional on the nodes of the subtree $T$. Indeed, 
$e(\widetilde{\Delta}) \geq e(\widetilde{\Delta}_0) + e(\widetilde{\Delta}_1),$
with $\widetilde{\Delta}_i$ being the children of $\widetilde{\Delta}$.
Applying the superadditivity recursively until only the leaves of the 
subtree \(T\) are left, we get
$$ e(\widetilde{\Delta})\geq \sum_{\Delta\in\calL(T)} e(\Delta) 
\geq \#\calL(T) a.
$$
\end{proof}

We can now show that, under the assumptions \eqref{eq:errorassumption1} and \eqref{eq:errorassumption2} on the error, the $h$-refinement strategy based on the greedy selection of the node with the largest error is near-best in the class of $h$-refined trees. This result resembles a corresponding statement proven in \cite{Binev18} which used a modified error functional instead of the maximal error.

\begin{theorem} \label{thm:h-nearbest}
Let \(\Tcal\) be a BGP tree, such that the local errors
$e(\Delta)$, \(\Delta\in\Tcal\),
satisfy the assumptions \eqref{eq:errorassumption1} and \eqref{eq:errorassumption2} for all 
internal nodes of \(\Tcal\).
Further, let $\calT_m$ with \(\#\Lcal(\Tcal_m)=m\) be the BGP subtrees of $\mathcal{T}$ recursively obtained by subdividing the leaf of $\calT_{m-1}$ with the maximum local error
$e(\Delta)$. Then, the tree $\calT_m$ provides a near-best $h$-adaptive
approximation error in the sense that
$$
\calE(\calT_m) \leq \frac{m}{c\,\max\{m-k+1,2\}} \sigma_k^h,$$
with 
$$\sigma_k^h \isdef \min_{T : \#\calL(T)\leq k} 
\sum_{\Delta\in\calL(T)} e(\Delta),$$
for any integer $k\leq m$, where $T$ denotes any BGP subtree of $\mathcal{T}$ with at most $k$ leaves. 
\end{theorem}

\begin{proof}
We define the sequence $$t_m \isdef \max_{\Delta\in\calL(\calT_m)} e(\Delta),$$
which, due to property \eqref{eq:errorassumption1} and the construction of the subtrees
$\calT_m$, is non-increasing with respect to $m$. Next, let $\calT_k^*$ be a BGP subtree of $\mathcal{T}$ with $k$ leaves for which best approximation is achieved, i.e.,
$\Ecal(\calT_k^*)=\sigma_k^h$. 

Comparing the two BGP subtrees $\calT_m$ and $\calT_k^*$ of $\calT$, a first possible scenario is that 
$\calT_k^* \subset \calT_m$. Due to the superadditivity \eqref{eq:errorassumption1} and the applied greedy
refinement, we get in this case that $\calE(\calT_m) \leq \calE(\calT_k^*)$. 
Conversely, if $\calT_m \subseteq \calT_k^*$ then $m \leq k$. Since by assumption $k \leq m$, this implies that $k = m$ and $\calT_m = \calT_k^*$, and hence $\calE(\calT_m) = \sigma_k^h$. This means that also in this second scenario the stated bound holds trivially. 

We may therefore assume that each of the two trees has nodes 
that do not belong to the other. Since a BGP tree with $m$ leaves has $2 m - 1$ nodes, the trees $\calT_m$ and $\calT_k^*$ have $2 m-1$ and $2 k-1$ nodes, respectively. Therefore, as $m \geq k$, there are at least $(2 m-1)-(2 k-1)=2(m-k)$ nodes of $\calT_m$ that are not nodes of $\calT_k^*$. The fact that $\calT_m$ and $\calT_k^*$ are full binary trees with the same root further implies that also in the case $k = m$ the tree $\calT_m$ contains at least $2$ nodes that are not in $\calT_k^*$. We denote the set of all such nodes by $B$. By our assumptions, $B$ is non-empty and contains at least $2 \max\{1, m - k\}$ nodes.
    
Since the subtree $\calT_m$ is built by greedy refinement,
every internal node of $\calT_m$ was chosen for refinement at a certain point, precisely when its local error was maximal. As the maximum leaf error $t_m$ is 
non-increasing, this implies that for every internal node $\Delta\in\calT_m$ the error satisfies $e(\Delta)\geq t_m$. Because of the additional assumption \eqref{eq:errorassumption2} on the distribution of the error among the children of a tree node, we therefore get 
$$e(\Delta)\geq c \, t_m \quad \text{ for all $\Delta \in \mathcal{T}_m$}.$$ 

Now, for any leaf $\Delta\in\calL(\calT_k^*)$, let us denote by $B_\Delta$ the nodes of $\calT_m$ that are descendant of $\Delta$, including $\Delta$
itself, and set $b_\Delta \isdef \#B_\Delta$. Clearly, if $\Delta$ is not contained in
$\calT_m$ then $B_\Delta$ is empty. Further, we have $b_\Delta\geq 1$ whenever
$\Delta$ is a node of $\calT_m$. Since every node of $B$ is a
descendant of a unique leaf of $\calT_k^*$, the family
$\{B_\Delta\}_{\Delta\in\calL(\calT_k^*)}$ is a partition of $B \cup (\calL(\calT_k^*) \cap \calT_m)$. 
If $b_\Delta\geq 1$, we can consider $B_{\Delta}$ as a full binary subtree of $\mathcal{T}_m$ rooted in $\Delta$ 
with $(b_\Delta+1)/2$ leaves.

Now, applying Lemma \ref{Lemma2} to the subtree $B_\Delta$, we obtain 
$$ 
e(\Delta) \geq \# \calL(B_{\Delta}) c\, t_m = \frac{b_\Delta+1}{2} c \, t_m.
$$
for all leaves $\Delta \in \calL(\calT_k^*) \cap \calT_m$. Thus,
\begin{align*}
    \sigma_k^h =& \sum_{\Delta\in\calL(\calT_k^*)} 
e(\Delta)\geq \sum_{\Delta \in \calL(\calT_k^*) \cap \calT_m} 
e(\Delta) \\ \geq& \sum_{\Delta\in \calL(\calT_k^*) \cap \calT_m} \frac{b_\Delta+1}{2} c \, t_m.
\end{align*}
Since
\begin{align*}
    \sum_{\Delta\in \calL(\calT_k^*) \cap \calT_m } (b_\Delta  +1)
    =& \#B + 2\# (\calL(\calT_k^*)\cap \calT_m) \\ \geq& \#B + 2 
    \geq  2 \max \{m - k , 1\} + 2,
\end{align*}
we therefore get
    $$\sigma_k^h \geq c\,\max\{m-k+1,2\} \, t_m.$$
Finally, we obtain for the total error $\calE(\calT_m)$ the estimate   
\begin{align*}
        \calE(\calT_m) &= \sum_{\Delta\in\calL(\calT_m)} e(\Delta) 
        \leq m \, t_m \leq \frac{m}{c\, \max\{m-k+1,2\}} \sigma_k^h.
    \end{align*}
\end{proof}

The result shows that greedy $h$-refinement with a fixed error quantity for each element of the partitioning tree can be near-best up to an explicit
multiplicative constant. However, the assumption \eqref{eq:errorassumption2} on the even distribution of the errors is not intrinsic for partition-based approximation methods, but rather a hypothesis
which excludes unbalanced error distributions after a split. The signal adaptivity of partitioning methods as the graph wedgelets leads to smaller combined errors on the children of a tree node, but it does not guarantee any balancedness between the errors. Thus, we can in general also not guarantee theoretically that the greedy $h$-refinement strategy of the wedgelets leads to a near-best $h$-refinement. On the other hand, the numerical tests in Section \ref{sec:num_test} suggest that the greedy strategy of the wedgelets does in general not perform that bad in practice. 

Nevertheless, we aim to find improved refinement strategies that include also adaptive polynomial refinements for the approximation errors. Indeed, as already suggested by the Jackson-type estimate in \eqref{eq:mtermalgebraicrate}, pure $h$-refinement is still limited to algebraic convergence
rates. This motivates the construction of $hp$-refinements introduced in the next section that are able to combine domain refinements with adaptively selected polynomial degrees for the local approximation.

\section{$hp\,$-adaptivity strategies} \label{sec:hp_strategies}
The convergence rate of piecewise polynomial approximation strongly depends on the regularity of the underlying signal. Pure $h$-refinement methods, which improve the approximation only through a finer partition of the domain,
are limited to algebraic convergence rates and cannot take full advantage of
additional smoothness. To achieve higher convergence rates, one can additionally exploit the smoothness of the signal by locally increasing the polynomial degree of the approximation. This leads to so-called $p$-adaptivity and, more generally, to $hp$-adaptive strategies, which combine local mesh refinement with variable polynomial degrees. It is 
well-established in approximation theory and finite element analysis that,
while $h$-refinements typically yield algebraic convergence, $p$-refinement can achieve exponential convergence rates for locally analytic functions, see, e.g., \cite{ApelMelenk17,BabuvskaSuri94,book:Schwab98}.

To allow $p$-adaptivity for graph domains, we assume that the underlying graph can be locally embedded in a domain $\Omega\subset\bbR^s$ in which polynomial basis systems are available. Once such an embedding is available, as detailed in Section \ref{sec:gen_graph}, it becomes possible to use $hp$-adaptive strategies on the given graph partitioning trees. The underlying idea is that in regions where the graph signal is sufficiently smooth high-order polynomial approximation can be used to obtain significantly higher approximation rates, and potentially even exponential convergence. 

Denoting by $\bbP_d(\bbR^s)$ the space of $s$-variate polynomials of total
degree at most $d$, i.e.,
\[ \bbP_d(\bbR^s)
= \operatorname{span}\bigl\{ \bs{x}^{\bs{\alpha}} \; : \;
\bs{\alpha}\in\bbN^s,\ \|\bs{\alpha}\|_1 \leq d \bigr\},\]
its dimension is
\[ \dim \bbP_d(\bbR^s) = \binom{d+s}{s} = \Ocal(d^s). \]
This implies that for embeddings in high-dimensional spaces $\bbR^s$ the cost
for the storage of the polynomial coefficients increases rapidly with the degree $d$. This cost has to be taken into account when comparing the efficiency of $p$-refinement against different levels of $h$-refinements. To measure this quantity, we define the complexity or cost $\calC (\calT,d)$ of an $hp$-approximation based on the tree $\calT$ decorated with the degrees $d$ as
\begin{equation}
    \calC (\calT,d) = \sum_{\Delta\in \calL(\calT)} c_{d(\Delta)}(\Delta),
\end{equation}
where $c_{d(\Delta)}(\Delta)$ denotes the number of coefficients used to represent 
the polynomial on a leaf $\Delta$ of the tree $\calT$. The complexity \(\calC (\calT,d)\) itself describes the total count of the coefficients used for the piecewise-polynomial approximation with respect to the leaf partition of the tree \(\Tcal\). The polynomial approximant of degree $d(\Delta)$ is concretely obtained by a least-squares fit of the signal values on the subdomains in $\calL(\calT)$. Herein, the fit by a constant mean value amounts to classical piecewise-constant signal approximation based on the partitions of a graph. This is related to a pure $h$-refinement for signal approximation and, for instance, used in the already mentioned work \cite{Erb23} on wedgelet approximation. 

For $hp$-refinement, polynomials with different degrees are employed to approximate the signal on the leaves of $\Tcal$. For this, we have to adapt the definition of the total error as
$$ 
\calE(\calT,d) \isdef \sum_{\Delta\in\calL(\calT)} e_{d(\Delta)}(\Delta).
$$
This generalizes the definition from the previous 
section in which the same local approximation method was applied to each partition element of the tree. In the following, we introduce now two strategies for $hp$-adaptivity on graphs.

\subsection{Greedy knapsack $hp$-adaptivity}
As a first strategy to steer the refinement process, we formulate $hp$-adaptivity as a binary knapsack problem, as for instance introduced in \cite{book:MartelloToth90}. The classical knapsack problem aims at maximizing a benefit under a prescribed capacity or cost budget. In our case, the benefits are reductions of the approximation error, while the costs account for the degrees of freedom introduced by the $h$- and $p$-refinements. 

Let $\calT$ be a fixed BGP tree rooted at the full vertex set $V$. We interpret $\calT$ as a reference tree. Inside this ambient tree all admissible adaptive trees $\calT'$ are subtrees $\calT' \subset \calT$, and, whenever a leaf of $\calT'$ has to be split, the children generated by the $h$-refinement are prescribed by $\calT$. Therefore, the geometry of the admissible splits is a priori fixed by $\calT$ and independent of the refinement criterion. Moreover, since $\calT$ is finite, it also limits the set of admissible refinements and provides an additional stopping criterion. One natural possibility is to choose $\calT$ as the complete BGP tree obtained by recursively applying a particular splitting rule until every leaf consists of a single vertex.

For each node $\Delta\in\calT$, we introduce binary decision
variables 
$$ x_\Delta^h, x_{\Delta,d}^p \in \{0,1\}.$$ 
Here, $x_\Delta^h = 1$ indicates that the node $\Delta$ is $h$-refined into its children $\Delta_0$ and $\Delta_1$, and $x_{\Delta,d}^p = 1$ indicates that there is a $p$-refinement on the node $\Delta$ from degree $d-1$ to degree $d$.

Each refinement is assigned a positive error reduction and a cost increase. 
In the case of $h$-refinement, we split a node $\Delta$ in two children $\Delta_0$, $\Delta_1$, and approximate the signal on each child by a degree-$0$ polynomial, i.e., a constant. Then, the error reduction and cost increase for this $h$-refinement are defined as
$$ E_h(\Delta) = e_0(\Delta) -\big(e_0(\Delta_0) + e_0(\Delta_1)\big), \quad W_h(\Delta)=1+\lambda,$$
where $e_0(\Delta)$ denotes the error of the best constant approximation on $\Delta$. The cost increase $W_h(\Delta)$ for the $h$-refinement is set equal to $1+\lambda$ since only one extra constant has to be stored for the approximation in the new partition and $\lambda$ corresponds to the cost of defining such a new partition. For a $p$-refinement, we set
$$ E_p(\Delta,d) = e_{d-1}(\Delta) - e_{d}(\Delta), \quad W_p(\Delta,d) = c_d - c_{d-1}, 
$$
where $e_{d}(\Delta)$ denotes the least-squares error of the best polynomial approximation of degree $d$, and $W_p(\Delta,d)$ the increase of polynomial coefficients when passing from degree $d-1$ to $d$. 

Given a maximal polynomial degree $d_{\max} \geq 1$ and a cost budget $C$, our objective is to find a configuration of the decision variables on the tree $\Tcal$ that minimizes the total approximation error without exceeding the prescribed degree $d_{\max}$ and the cost budget $C$. This is equivalent to maximize the following cumulative error reduction: 
\begin{equation}
\begin{aligned}
\max_{x_\Delta^h,x_{\Delta,d}^p} \quad & 
\sum_{\Delta\in\calT} \bigg(E_h(\Delta)x_\Delta^h + 
\sum_{d=1}^{d_{\max}} E_p(\Delta,d) x_{\Delta,d}^p \bigg) \\
\textrm{s.t.} \quad & \sum_{\Delta\in\calT} 
\bigg(W_h(\Delta)x_\Delta^h + \sum_{d=1}^{d_{\max}} W_p(\Delta,d) 
x_{\Delta,d}^p \bigg) \leq C, \\
                    & x_\Delta^h,x_{\Delta,d}^p\in\{0,1\}
                    \quad\text{for all } \Delta\in\calT.    \\
\end{aligned}
\end{equation}

Additional admissibility constraints ensure that the selected variables form a valid configuration for an $hp$-adaptive tree. 
First, a child $\Delta_i$ can only be $h$-refined if its parent $\Delta$ has already been refined, i.e., 
\[ x_{\Delta_i}^h \leq x_\Delta^h, \qquad i \in \{0,1\}.\]
Since $p$-refinement is restricted to the leaves of the adaptive tree, we also require
$$ x_{\Delta_i,d}^p \leq x_{\Delta}^h - x_{\Delta_i}^h, \quad i \in \{0,1\}.$$
Moreover, polynomial degrees can only be incorporated sequentially in the $p$-refinement process, so that
$$ x_{\Delta_i,d}^p \leq x_{\Delta_i,d-1}^p,\quad i \in \{0,1\}.$$

Rather than solving this binary optimization problem exactly, we employ an efficiency-based greedy strategy motivated by a linear relaxation of the knapsack problem and Dantzig's algorithm \cite{book:MartelloToth90}. While Dantzig's algorithm gives an exact optimal solution for the fractional knapsack problem, it provides only a heuristic for the binary problem with the tree constraints described above. 

For our greedy knapsack algorithm, we associate with each admissible refinement an efficiency measure defined as the following error reduction per unit cost:
$$ r_\Delta^h \isdef \frac{E_h(\Delta)}{W_h(\Delta)},
\quad r_\Delta^p(d) \isdef \frac{E_p(\Delta,d)}{W_p(\Delta,d)}.$$
Then, at each iteration of the algorithm, we greedily select among all admissible actions on the current leaves the one with maximal efficiency, and update the adaptive tree accordingly.  

The implementation of this strategy in the complementary form, in which the cost is minimized subject to a certain error tolerance, is summarized in Algorithm \ref{alg:knapsack}. A node $\Delta \in \calT$ is called $h$-refinable if it is an internal node of the ambient tree $\Tcal$ and $p$-refinable if its current degree satisfies $d(\Delta) < d_{\max}$. Starting from an initial subtree $\calT_K \subset \calT$ (which could be just the root of $\mathcal{T}$) the algorithm iteratively performs the most efficient admissible refinement until a prescribed tolerance is reached, or no further refinement is possible. In particular, if the reference tree $\calT$ is only a truncated variant of the complete BGP tree, the prescribed tolerance might not be reached and the algorithm terminates with the reference tree $\calT$ as an output.  

\begin{algorithm}[htb]
\caption{Greedy knapsack $hp$-adaptivity}
\label{alg:knapsack}

\textbf{Input:} reference tree $\mathcal{T}$, initial tree
$\mathcal{T}_{K}\subseteq\mathcal{T}$, tolerance
$\texttt{tol}$, maximal degree $d_{\max} \geq 1$.

\vspace{1mm}

\For{each leaf $\Delta\in\mathcal{L}(\mathcal{T}_{K})$}{
    set $d(\Delta) \coloneqq  0$ 
    and compute the $p$-refinement efficiency $r_\Delta^p \coloneqq r_\Delta^p(1)$\;

    \If{$\Delta$ is $h$-refinable in $\mathcal{T}$}{
        compute the $h$-refinement efficiency
        $r_\Delta^h$ \;
    }
}

\While{$\mathcal{E}(\mathcal{T}_k)>\textnormal{\texttt{tol}}$ and an $h$-refinable or $p$-refinable leaf node exists}{

    select
    \[
    (\overline\Delta,\overline\rho)
    \in
    \arg\max
    \bigl\{
        r_\Delta^\rho :
        \Delta\in\mathcal L(\mathcal T_{K}),
        \ \Delta\text{ is }\rho\text{-refinable}, \ \rho\in\{h,p\}
    \bigr\};
    \]

    \eIf{
        $\overline\rho=p$
    }{
        perform a $p$-refinement of
        $\overline{\Delta}$ and set $d(\overline{\Delta}) := d(\overline{\Delta})+1$;
        

        \If{$\overline\Delta$ is still $p$-refinable}{
            compute
            $
            r_{\overline\Delta}^p
            :=
            r_{\overline\Delta}^p
            \bigl(d(\overline\Delta)+1\bigr)
            $ \;
        }
    }{
        perform an $h$-refinement of $\overline{\Delta}$ and add the children 
        $\overline{\Delta}_0$, $\overline{\Delta}_1$ of $\overline{\Delta}$ to the tree $\mathcal{T}_{K}$\;
        set $d(\overline{\Delta})=0$;

        \For{$i\in\{0,1\}$}{
            set $d(\overline{\Delta}_i) \coloneqq 0$ 
            and compute the $p$-refinement efficiency $r_{\overline{\Delta}_i}^p := r_{\overline{\Delta}_i}^p(1)$\;

        \If{$\overline{\Delta}_i$ is $h$-refinable}{
             compute the $h$-refinement efficiency
             $r_{\overline{\Delta}_i}^h$ \;
         }}
    }
    update the global approximation error
    $\mathcal{E}(\mathcal{T}_{K},d)$\;
}

\vspace{1mm}

\textbf{Output:}
the subtree $\Tcal_{K}$ and the degrees $d(\Delta)$ associated with each
leaf $\Delta\in\mathcal L(\mathcal T_{K})$.

\end{algorithm}

\subsection{Error-cost pruning}
In this section, we consider a complementary $hp$-adaptivity strategy 
inspired by the tree construction proposed in \cite{Binev18}. 
This second strategy for $hp$-adaptivity is based on pruning an already refined
tree through a bottom-up approach where we replace the approximation at the
leaves with a higher-degree polynomial at an ancestor node if it is more 
efficient. To this end, we pursue a two-phase strategy which consists in a first
expansion phase and a subsequent pruning phase. The implementation of this
procedure is described in Algorithm \ref{alg:pruning}. 

\begin{algorithm}[htb]
\caption{Adaptive error cost pruning (ECP)}
\label{alg:pruning}

\textbf{Input:}
reference tree $\mathcal \calT$, initial tree
$\mathcal \calT_{ECP} \subseteq\calT$, tolerance $\texttt{tol}$,
maximal degree $d_{\max}$.

\vspace{1mm}

Set $d(\Delta):=0$ for any $\Delta\in \calT_{ECP}$ and compute
\[
    \Ecal(\calT_{ECP},d)
    =
    \sum_{\Delta\in\mathcal L(\calT_{ECP})}
    e_{d(\Delta)}(\Delta);
\]

\While{
    $\Ecal(\calT_{ECP},d)>\textnormal{\texttt{tol}}$
    and an $h$-refinable leaf exists
}{
    select an $h$-refinable leaf
    $\overline{\Delta}\in\mathcal L(\calT_{ECP})$\;

    add the children
    $\overline{\Delta}_0,\overline{\Delta}_1 \in \calT$ of $\overline{\Delta}$ to the tree $\calT_{ECP}$ and set
    \[
        d(\overline{\Delta}_0)
        :=
        d(\overline{\Delta}_1)
        :=
        0;
    \]

    set $\Delta:=\overline{\Delta}$\;

    \While{$\Delta \neq \emptyset$ }{
        \If{ $d(\Delta)<d_{\max}$ and
            $\dim\mathbb P_{d(\Delta)+1}(\mathbb R^s)
                \leq
                \#\mathcal L(T_\Delta)$
        }{
            set
                $d(\Delta):=d(\Delta)+1$ \;
        }

        \eIf{$\Delta$ is the root of $\calT_{ECP}$}{
            set $\Delta = \emptyset$\;
        }{
            replace $\Delta$ by its parent\;
        }
    }

    update the error
    \[
        \mathcal E(\calT_{ECP},d)
        :=
        \sum_{\Delta\in\mathcal L(\calT_{ECP})}
        e_{d(\Delta)}(\Delta);
    \]
}

\vspace{2mm}

Visit the internal (non-leaf) nodes $\Delta\in\calT_{ECP}$ bottom-up
towards the root \;

\ForEach{visited internal node $\Delta$}{
    \If{ 
        $
            \displaystyle e_{d(\Delta)}(\Delta)
            \leq
            \sum_{\nabla\in\mathcal L(T_\Delta)}
            e_{d(\nabla)}(\nabla)
        $
    }{
        replace the subtree $T_\Delta$ rooted at $\Delta$ by the leaf $\Delta$\;
    }
}

\vspace{2mm}

\textbf{Output:}
The tree $\calT_{ECP}$, and the degrees $d(\Delta)$ associated with
$\Delta\in \calT_{ECP}$.
\end{algorithm}


All nodes are initially equipped with constant approximations. The first phase then performs successive $h$-refinements until the 
tree error falls below a given threshold. After each $h$-refinement, the polynomial degrees on the tree nodes along the path from the refined leaf to the root are updated.

More precisely, in each iteration, a selected leaf $\overline{\Delta}$ of the current tree $\calT_N$ is split to form a new tree $\calT_{N+1}$ with the children of $\overline{\Delta}$ added. Then, we traverse the path from 
$\overline{\Delta}$ back to the root of $\calT_{N+1}$ and update the polynomial degree of each ancestor $\Delta$ according to the number of leaves of the subtree $T_{\Delta}$ rooted in $\Delta$. If $d(\Delta)$ denotes the current degree associated to $\Delta$, 
it is increased to $d(\Delta)+1$ if 
\(c_{d(\Delta)+1}(\Delta)\leq \#\calL(T_\Delta)\) and \(d(\Delta)<d_{\max}\).
Thus, the polynomial degree at $\Delta$ is increased whenever the corresponding polynomial approximation requires fewer coefficients than the number of leafs in the subtree $T_{\Delta}$ rooted in $\Delta$.

After this expansion phase, we obtain a refined tree whose internal 
nodes are equipped with cost-efficient polynomial degrees. In the second phase, the tree is then traversed bottom-up and subtrees are pruned whenever the polynomial approximation on a node $\Delta$ is at least as accurate as the combined approximations on the leaves of $T_{\Delta}$. In particular, the descendants of a node $\Delta$ are removed from the adaptive tree if 
\begin{equation} \label{bound_leaves-ECP}
    e_{d(\Delta)}(\Delta) \leq 
    \sum_{\nabla\in\calL(T_\Delta)} e_{d(\nabla)}(\nabla).
\end{equation} 
In this case, the approximation linked to the subtree $T_{\Delta}$ is replaced by a single approximation on $\Delta$ without increasing the approximation error.  

The following result links the error and cost of the tree obtained after the initial expansion phase to the tree obtained after applying the error-cost pruning strategy in the second phase of Algorithm \ref{alg:pruning}.

\begin{theorem}
\label{thm:hp_pruning_complexity}
Let $\calT$ be a tree with total cost $\calC(\calT,d)$, constructed by the
expansion phase of Algorithm~\ref{alg:pruning}, and let
$\calT_{ECP}$ be the tree obtained after the subsequent error-cost pruning step. Then,
\[
    \calE(\calT_{ECP},d)
    \leq
    \calE(\calT,d),
\]
and
\[
    \calC(\calT_{ECP},d)
    \leq
    \calC(\calT,d).
\]
\end{theorem}

\begin{proof}
We note that Algorithm \ref{alg:pruning} constructs 
$\calT_{\mathrm{ECP}}$ from the expanded tree $\calT$ by traversing it from the leaves to the root by pruning a subtree $T_\Delta$ if 
    $$e_{d(\Delta)}(\Delta) 
    \leq \sum_{\nabla \in \calL(T_\Delta)}e_{d(\nabla)}(\nabla).$$
Moreover, by construction, we always have
    $$c_{d(\Delta)}(\Delta) \leq \# \calL(T_\Delta).$$
Hence, we get that 
\begin{align*}
  \calC(\calT_{ECP},d) 
=& \sum_{\Delta\in\calL(\calT_{ECP})}
c_{d(\Delta)} \\ \leq& \sum_{\Delta\in\calL(\calT_{ECP})}\#\calL(T_\Delta) 
    = \sum_{\Delta\in\calL(\calT)} 1 
    \\ \leq& \sum_{\Delta\in\calL(\calT)} c_{d(\Delta)} = \calC(\calT,d),  
\end{align*}
    as well as
    \begin{align*}
    \calE(\calT_{ECP},d) &= \!\!\! 
    \sum_{\Delta\in\calL(\calT_{ECP})} \!\!\!
    e_{d(\Delta)}(\Delta,d) \\ \leq& \!\!\! \sum_{\Delta\in\calL(\calT_{ECP})} 
    \sum_{\nabla\in\calL(T_\Delta)} e_{d(\nabla)}(\nabla) 
    \\ =& \!\!\! \sum_{\Delta\in\calL(\calT)} e_{d(\Delta)}(\Delta) = \calE(\calT,d).
    \end{align*}

\end{proof}

The ECP phase in Algorithm \ref{alg:pruning} thus ensures that the polynomial approximation error and the respective cost associated with each leaf of the pruned tree do not exceed the error and the cost of the expanded subtree it replaces. 

The ECP $hp$-refinement strategy of Algorithm \ref{alg:pruning} can be considered as complementary to the greedy knapsack approach, which selects between $h$- and $p$-refinements based on their respective efficiencies. In particular, the expansion phase in Algorithm \ref{alg:pruning} can be used to further refine branches of an initial tree by $h$-refinements and to include higher-order polynomial approximations at the internal nodes. While the expansion phase can increase the cost relative to an initial knapsack tree, it provides additional flexibility to reduce the approximation error. The subsequent pruning phase then removes subtrees whenever the respective approximation can be replaced by a higher-degree polynomial approximation with no larger error and no greater cost. In the wedgelet framework, this pruning step provides an additional storage benefit. A reduced number of leaves leads to a respectively reduced number of landmarks that need to be stored for the representation of the wedgelet tree.

\section{Polynomial function spaces for general graphs} 
\label{sec:gen_graph}
The two $hp$-refinement strategies for graph signal approximation introduced in the previous section assume 
that the underlying graph can be embedded into a Euclidean domain 
$\Omega\subset\bbR^s$. This embedding allows the usage of polynomial basis systems for local signal approximation. If the embedding is approximately isometric, local graph distances can moreover be related to respective Euclidean distances in $\Omega$. 

For general graphs, however, a single low-distortion Euclidean embedding may not exist. We therefore adopt a patch-wise local embedding strategy similar to the one introduced in \cite{EGMQ25}. For this, we partition the graph into $\rho$ mutually disjoint subgraphs $U_r$, $r\in \{1,\dots,\rho\}$, representing meaningful local neighborhoods. To identify these sub-graphs, we employ standard graph partitioning algorithms based on multilevel heavy edge matching, as implemented in \texttt{METIS} \cite{KK97,KK98}. 
Alternative methods, such as recursive bisection via spectral clustering \cite{vLux07}, could be used as well. If the graph vertices are samples of an unknown manifold, this decomposition mimics an atlas, with each patch playing the role of a discrete local chart.

Each patch is equipped with a low-dimensional coordinate map
$$\phi_r\colon U_r\to\bbR^m, $$
which provides the coordinates needed to define polynomial spaces for the $p$-refinements. Such embeddings can be constructed using multidimensional scaling \cite{book:coxcox,book:LeeVerleysen}, or, in the presence of an
underlying manifold structure, more specialized methods such as \texttt{ISOMAP},
see \cite{TdSL00}, or its computationally efficient variant \texttt{L-ISOMAP}
\cite{ST02}. Other dimensionality-reduction techniques such as
 \texttt{t-SNE} \cite{vdMH08} and \texttt{UMAP} \cite{MHM20} may also be applied. 

The embeddings naturally induce polynomial functions on each graph
patch. Given a multivariate monomial 
$${\bs x}^{\bs \alpha} = x_1^{\alpha_1} x_2^{\alpha_2}\cdots x_m^{\alpha_m}, \quad {\bs \alpha} = (\alpha_1,\dots,\alpha_m)\in\bbN^m, $$
we define the corresponding monomial function on the patch $U_r$ as
$$ {\bs X}^{\bs \alpha} = {\bs x}^{\bs \alpha} \circ \phi_r, $$
so that ${\bs X}^{\bs \alpha}(\node{v}) = \phi_r(\node{v})^{\bs \alpha}$ for $\node{v} \in U_r$. 
A separate BGP tree for $hp$-refinement is then constructed on each patch, providing in fact a forest for the approximation of the graph signal.

The embedding dimension $m$ directly controls the approximation
cost. For a tree node $\Delta$ associated with a patch $U_r$, the number of coefficients required for a polynomial approximation of degree $d(\Delta)$ depends on the dimension $m$ of the local embedding space. Hence, a selection $m\ll s$, can considerably mitigate the curse of dimensionality and improve the efficiency of both refinement strategies in $hp$-adaptivity. 

Moreover, for low dimensions $m$ the cost $W_p(\Delta,d)$ 
in the greedy knapsack approach of Algorithm \ref{alg:knapsack} becomes smaller and thus increases the efficiency ratio $r^p_\Delta$ for the $p$-refinement.
Similarly, in the ECP approach of Algorithm \ref{alg:pruning}, a smaller polynomial cost makes the bound $c_{d(\Delta)+1}(\Delta)\leq \#\calL(T_\Delta)$, imposed on the number of leaves of the sub-tree rooted at $\Delta$, easier to be satisfied.

\section{Numerical Tests} \label{sec:num_test}
In this section, we evaluate the performance of the proposed $hp$-adaptivity 
strategies on various discrete multidimensional signals. We test the two
$hp$-refinement steps separately as well as in a combined way. The numerical experiments are conducted on a
standard RGB image, the ERA5 climate dataset \cite{ERA5}, and a 3D MRI scan from the
UCSD-PTGBM dataset \cite{MRIdataset}. These different datasets allow us to 
thoroughly assess the robustness of our approaches across different dimensions 
and signals. We compare the introduced methods with the near-optimal $hp$-tree construction introduced in \cite{Binev18}, adapted here so that only the polynomial degree is varied instead of enriching the polynomial basis by one element at each refinement step. As a reference tree, we use a BGP tree generated in terms of a fully-adaptive (FA) greedy wedge splitting, as described in \cite{Erb23}. To reduce the cost of the fully adaptive wedge splits, the initial tree $\calT_K$ is obtained by $L$ levels of uniform binary pre-partitioning, bisecting each cell along the longest side of its bounding box (at the median for ERA5, so that both children have the same cardinality). We use $L=13$ ($2^{13}$ cells) for the MRI and the image, $L=11$ for ERA5, and $L=8$ within each patch of the Swiss roll. These initial cells start as leaves with a constant approximation.

Due to the fact that in the wedgelet setting, once we partition we introduce a new center that need to be store, we simplify and set $\lambda=1$ as cost for storing the partition.

Through the numerical experiments, we validate the quality of the approximation 
$\tilde{f}$ of an $m$-variate discrete signal $f\in\bbR^{N\times m}$ by the 
column-normalised mean-square error
$$ 
\Ecal(\Tcal) = \frac{1}{N} \| (\tilde{f}-f)D_m\|^2_F, 
$$
with $D_m 
= \mathrm{diag} (s_1,\dots,s_m)$, $s_j 
= \frac{1}{\max_{1\leq i \leq N} |f_{ij}|}$, and where $\|\cdot \|_F$ denotes the Frobenius norm. We employ the rescaling $D_m$ such
that $\Ecal(\Tcal)$ is dimensionless and a unique threshold can be used across
different signals and patches. In all experiments, we compare the number of
coefficients $\calC(\Tcal,d)$ (corresponding to the cost), and the number of leaves $\#\calL(\Tcal)$ (corresponding to the number of landmarks in the wedgelet regime), employed to reach a total error smaller or equal to \texttt{tol}$=10^{-4}$.

In the subsequent tables, the first column ($h$-max) contains the values associated to a pure greedy $h$-refinement in which the leaf with the maximum error is refined by a fully-adaptive wedge split \cite{Erb23}. The
second column ($h$-Binev) follows Binev's original strategy for $h$-refinement \cite{Binev18}. The third column ($hp$-K) gives the values for the
$hp$-adaptive tree obtained via the greedy knapsack algorithm. The fourth column
($hp$-ECP) reports the values obtained by applying the ECP strategy to the $h$-refinement tree of the first column.   
The fifth column ($hp$-K+ECP) combines the knapsack method in Algorithm \ref{alg:knapsack} with the error-cost pruning in Algorithm \ref{alg:pruning}. Finally, the last column ($hp$-Binev) reports the respective values obtained with Binev's near-optimal $hp$-strategy described in \cite{Binev18}.

In all tables, the most favourable values are highlighted in
bold. For all methods involving polynomial refinement, the maximum polynomial degree is set to $d_{\max} = 5$.

The first dataset consists of a 3D MRI scan with resolution
$256 \times 256 \times 256$. This corresponds to $16\,777\,216$ node points.
Three slice projections of this dataset in the axial, coronal, and sagittal directions,
from left to right, are depicted in Figure~\ref{fig:MRI}. In
Table~\ref{tab:complex_MRI}, we report the respective cost (number of polynomial coefficients) and the number of leaves (number of landmarks) for the obtained trees. Combining both numbers as relevant storage costs, we observe that applying only the
greedy knapsack criterion reduces these by $15.7\%$, whereas applying
only error-cost pruning yields an improvement of $18.9\%$. When
combining the two methods, the reduction increases to $25.7\%$, which
outperforms the $15.5\%$ improvement obtained with Binev's
$hp$-near-optimal construction. 
\begin{table}[htb]
\centering
\begin{tabular}{c|c|c|c|c|c|c}
 & $h$-max & $h$-Binev & $hp$-K & $hp$-ECP & $hp$-K+ECP & $hp$-Binev \\
\hline
    $\calC(\calT)$ & $74\,487$ & $76\,419$ & $73\,570$ & $\bm{66\,094}$ & 
    $71\,186$ & $69\,075$ \\
\hline
$\#\calL(\calT)$ & $74\,487$ & $76\,419$ & $52\,032$ & $54\,731$ & 
$\bm{39\,523}$ & $56\,882$ \\
\hline
$\calC(\calT)+\#\calL(\calT)$ & $148\,974$ & $152\,838$ & $125\,602$ & $120\,825$ & $\bm{110\,709}$ & $125\,957$
\end{tabular}
\caption{Comparison of complexity (total number of polynomial coefficients) and leaf 
cardinalities (number of landmarks) for various $h$- and $hp$-adaptive approximation 
strategies applied to the 3D MRI scan.}
\label{tab:complex_MRI}
\end{table}

\begin{figure}[htb]
\centering
\begin{tikzpicture}
	\draw(-5,0) node{\includegraphics[scale = 0.5, clip]{%
    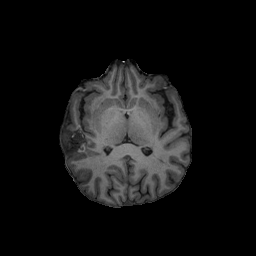}};
    \draw(0,0) node{\includegraphics[scale = 0.5, clip]{%
      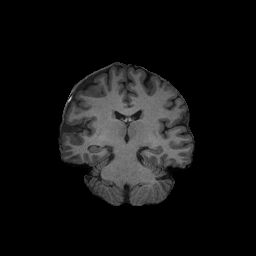}};
    \draw(5,0) node{\includegraphics[scale = 0.5, clip]{%
      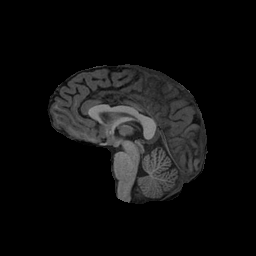}};
\end{tikzpicture}
\caption{Three slice projections (axial, coronal, and sagittal) of a 
  $256 \times 256 \times 256$ 3D MRI scan from the UCSD-PTGBM dataset, 
  used for evaluating the $hp$-adaptive refinement strategies.
} \label{fig:MRI}
\end{figure}

For the ERA5 dataset, we preprocess the data by resampling it via an
interpolation procedure at the Fibonacci lattice points on the unit sphere. This is done to avoid clustering near the poles caused by the original cylindrical
representation of the data. We consider a total of $1\,038\,240$ points. 
The graph signal corresponds to the 2-meter temperature, in
degrees Celsius, measured at midday on June 24, 2024, see
Figure~\ref{fig:ERA5}.

We report the approximation cost and the number of tree leaves in
Table~\ref{tab:complex_ERA}. In terms of storage, the greedy knapsack
criterion and ECP yield improvements of $11.3\%$ and
$15\%$, respectively. Their combined application results in a $19.0\%$
reduction, exceeding the $12.1\%$ gain achieved with Binev's
$hp$-near-optimal construction.

Figure~\ref{fig:poly_ERA} depicts the polynomial degrees used in each
patch induced by the tree partition.

\begin{table}[htb]
\centering
\begin{tabular}{c|c|c|c|c|c|c}
 & $h$-max & $h$-Binev & $hp$-K & $hp$-ECP & $hp$-K+ECP & $hp$-Binev \\
\hline
 $\calC(\calT)$ & $13\,219$  & $13\,227$  & $13\,062$  & $\bm{12\,750}$ & 
 $13\,719$ & $13\,132$ \\
\hline
$\#\calL(\calT)$ &  $13\,219$ & $13\,227$ & $10\,380$ & $9\,718$ &
$\bm{7\,703}$ & $10\,097$ \\
\hline
$\calC(\calT)+\#\calL(\calT)$ & $26\,438$ & $26\,454$ & $23\,442$ & $22\,468$ &
$\bm{21\,422}$ & $23\,229$ \\ 
\end{tabular}
\caption{Comparison of complexity (total number of polynomial coefficients) and leaf 
cardinalities (number of landmarks) for various $h$- and $hp$-adaptive 
approximation strategies applied to the ERA5 dataset.}
\label{tab:complex_ERA}
\end{table}

\begin{figure}[htb]
\centering
\begin{tikzpicture}
	\draw(0.5,0) node{\includegraphics[scale = 0.055,clip,trim = 750 0 750 0]
    {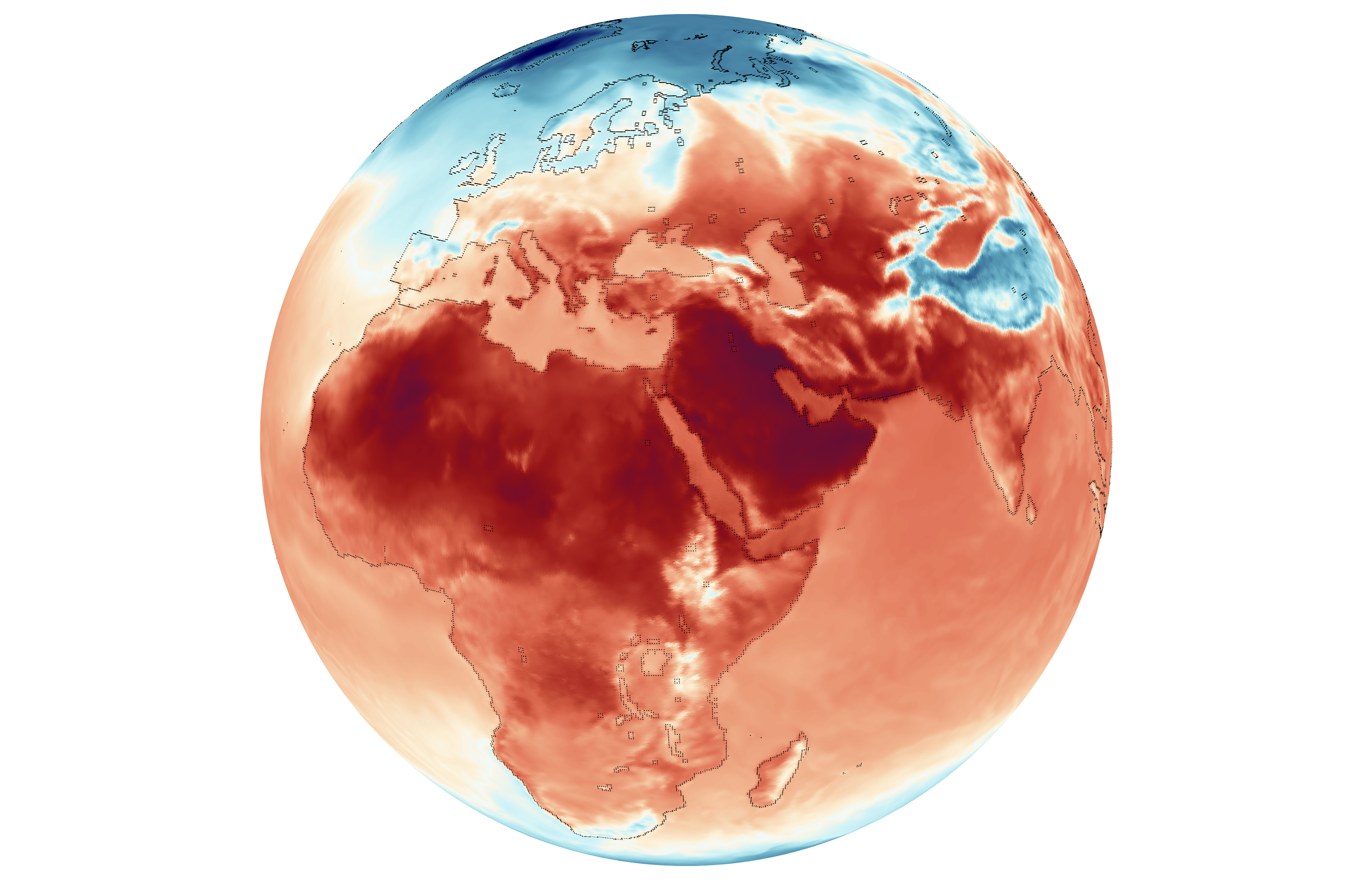}};
\end{tikzpicture}
\caption{Visualization of the discrete signal on the ERA5 climate dataset,
  representing the 2-meter level temperature in degrees Celsius at midday
  on June 24, 2024.
} \label{fig:ERA5}
\end{figure}

\begin{figure}[htb]
\centering
\begin{tikzpicture}
	\draw(-3.5,3.5) node{%
    \includegraphics[scale = 0.17, clip,trim = 850 300 850 100]{%
  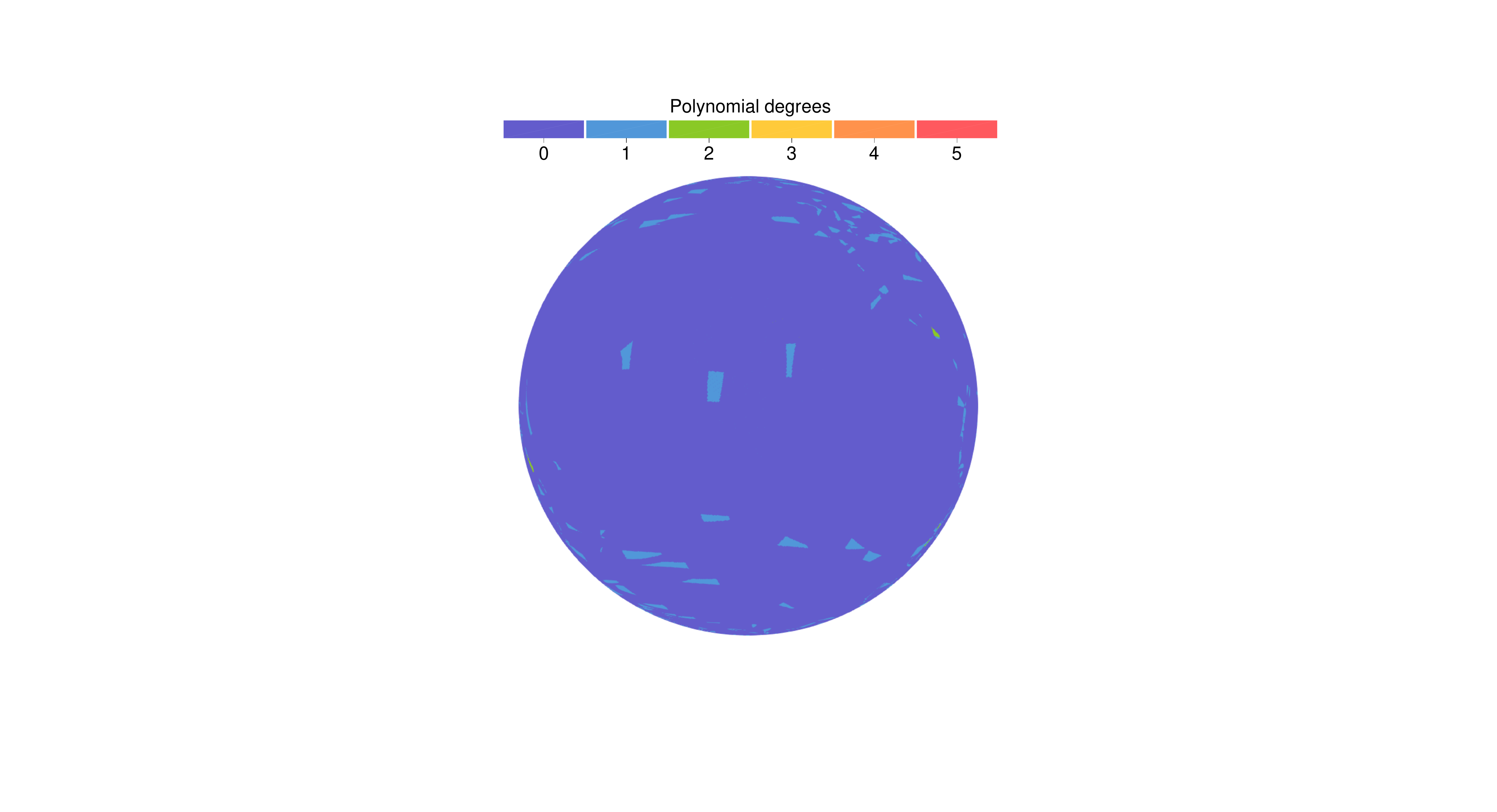}};
    \draw(3.5,3.5) node{%
      \includegraphics[scale = 0.17, clip,trim = 850 300 850 100]{%
    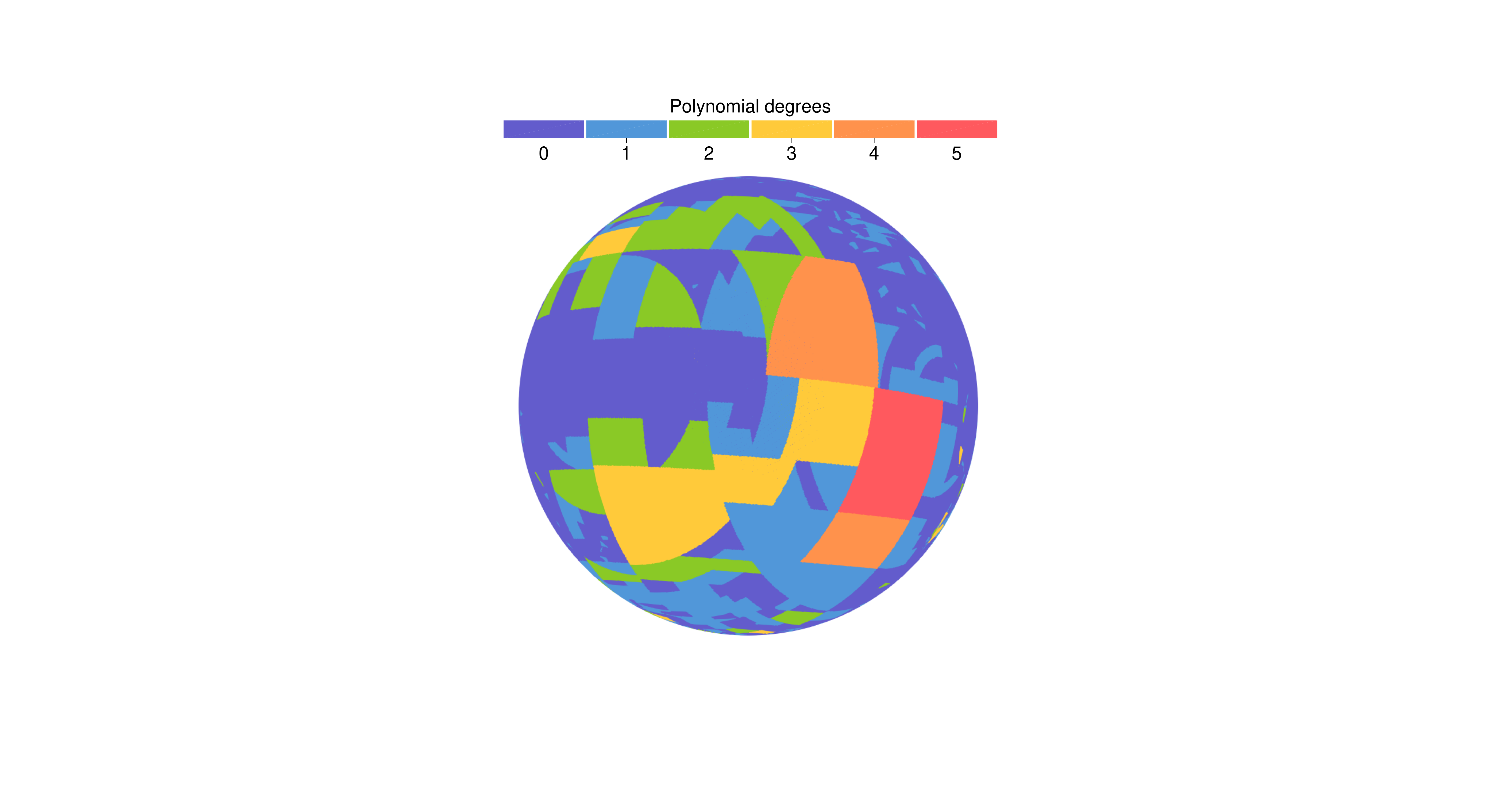}};
    \draw(-3.5,-3.5) node{%
      \includegraphics[scale = 0.17, clip,trim = 850 300 850 100]{%
    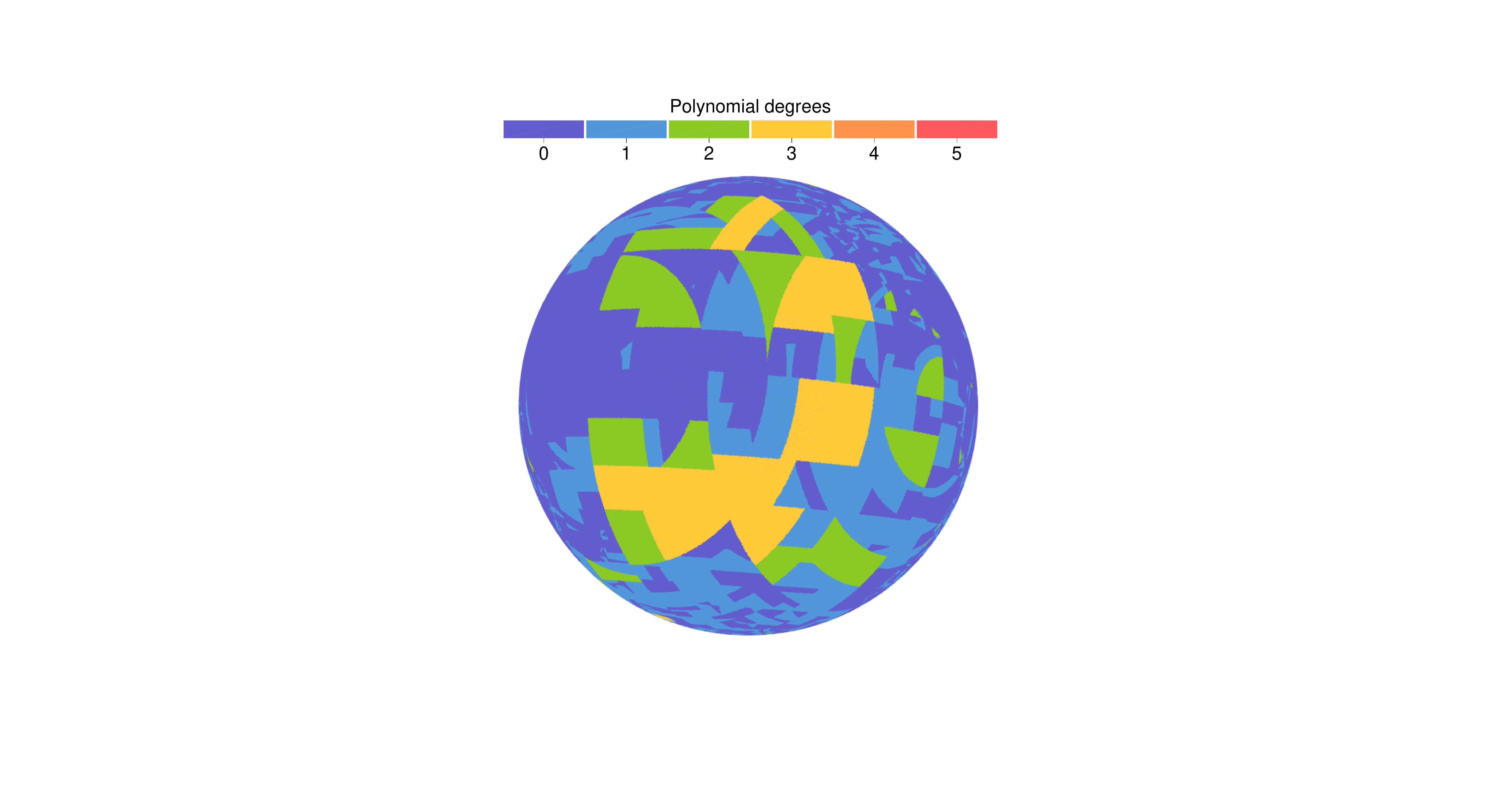}};
    \draw(3.5,-3.5) node{%
      \includegraphics[scale = 0.17, clip,trim = 850 300 850 100]{%
    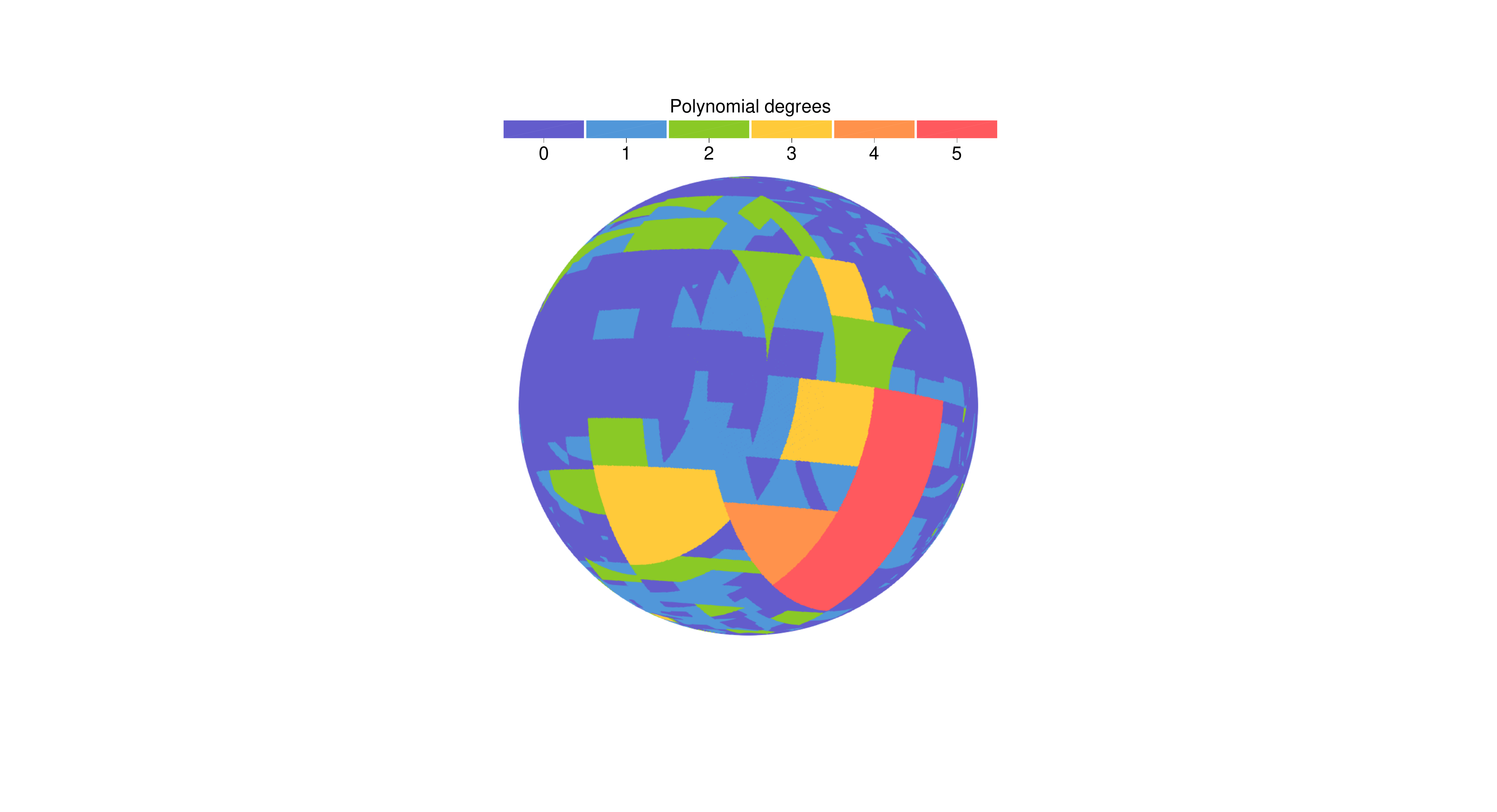}};
\end{tikzpicture}
\caption{Spatial distribution of the polynomial degrees employed in each tree
  partition patch for the ERA5 dataset across four $hp$-adaptive strategies: 
  $hp$-Knapsack (top-left), $hp$-ECP (top-right), $hp$-K+ECP (bottom-left)
  and $hp$-Binev (bottom-right).
} \label{fig:poly_ERA}
\end{figure}

The following test considers an image of the Tower of Bel\'em in Lisbon, see
Figure~\ref{fig:Belem}. The image has a resolution of
$4080 \times 3072$, corresponding to $12\,533\,760$ points per color
channel.

Table~\ref{tab:complex_Belem} reports the cost and the number of
tree leaves for the individual RGB channels. The greedy knapsack
criterion and the ECP method yield storage reductions of
approximately $11.8\%$ and $9.3\%$, respectively, across the three
channels. Specifically, the reductions for the greedy knapsack are
$11.4\%$ (R), $12.3\%$ (G), and $11.6\%$ (B), while for error-cost pruning
they are $9.4\%$ (R), $9.5\%$ (G), and $9.0\%$ (B).

Their combined application results in reductions of $14.4\%$ (R),
$15.3\%$ (G), and $14.5\%$ (B), yielding a total reduction of $14.7\%$.
In comparison, Binev's near-optimal $hp$-refinement achieves a total gain of $2.0\%$ and channel-wise improvements of $2.2\%$ (R), $2.2\%$ (G), and $1.8\%$ (B), respectively.

Figure~\ref{fig:poly_Belem} depicts the polynomial degrees used in each
patch induced by the tree partition for the red channel.
\begin{table}[htb]
\centering
\small
\begin{tabular}{c|c|c|c|c|c|c}
 & $h$-max &  $h$-Binev & $hp$-K & $hp$-ECP & $hp$-K+ECP & $hp$-Binev \\
\hline
 $\calC(\calT_{red})$  &  $788\,837$ & $798\,733$  &  $789\,144$ & 
 $\bm{776\,479}$ & $864\,324$ & $844\,297$ \\
\hline
$\#\calL(\calT_{red})$  & $788\,837$ & $798\,733$ & $608\,599$ & 
$653\,455$ & $\bm{486\,532}$ & $699\,856$ \\
\hline
$\calC(\calT_{red})+\#\calL(\calT_{red})$  & $1\,577\,674$ & $1\,597\,466$ & $1\,397\,743$ & 
$1\,429\,934$ & $\bm{1\,350\,856}$ & $1\,544\,153$ \\
\hline
 $\calC(\calT_{green})$ & $779\,071$  & $785\,898$  &  $773\,144$ & 
 $\bm{766\,616}$ & $845\,284$ & $833\,406$ \\
\hline
$\#\calL(\calT_{green})$ & $779\,071$ & $785\,898$ & $593\,047$ & 
$643\,564$ & $\bm{474\,667}$ & $691\,108$ \\
\hline
$\calC(\calT_{green})+\#\calL(\calT_{green})$  & $1\,558\,142$ & $1\,597\,466$ & $1\,366\,191$ & 
$1\,410\,180$ & $\bm{1\,319\,951}$ & $1\,524\,514$ \\
\hline
 $\calC(\calT_{blue})$  & $771\,179$ & $778\,794$ & $768\,388$  & 
 $\bm{759\,737}$ & $841\,469$ & $825\,331$ \\
\hline
$\#\calL(\calT_{blue})$ & $771\,179$ & $778\,794$ & $594\,517$ & 
$643\,765$ & $\bm{477\,499}$ & $689\,746$ \\
\hline
$\calC(\calT_{blue})+\#\calL(\calT_{blue})$  & $1\,542\,358$ & $1\,557\,588$ & $1\,362\,905$ & 
$1\,403\,502$ & $\bm{1\,318\,968}$ & $1\,515\,077$ \\
\hline
\hline
 $\calC(\calT_{total})$  &  $2\,339\,087$ & $2\,363\,425$  & 
 $2\,330\,676$ & $\bm{2\,302\,832}$  & $2\,551\,077$ & $2\,503\,034$ \\
\hline
$\#\calL(\calT_{total})$ & $2\,339\,087$ & $2\,363\,425$ & 
$1\,796\,163$ & $1\,940\,784$ & $\bm{1\,438\,698}$ &  $2\,080\,710$ \\
\hline
$\calC(\calT_{total})+\#\calL(\calT_{total})$  & $4\,678\,174$ & $4\,726\,850$ & $4\,126\,839$ & 
$4\,243\,616$ & $\bm{3\,989\,775}$ & $4\,583\,744$ \\ 
\end{tabular}
\caption{Complexities and number of leaves for the Tower of Bel{\'e}m image,
evaluated across individual RGB channels and in total.}
\label{tab:complex_Belem}
\end{table}

\begin{figure}[htb]
\centering
\begin{tikzpicture}
	\draw(-4.8,0) node{\includegraphics[scale = 0.075,clip]
    {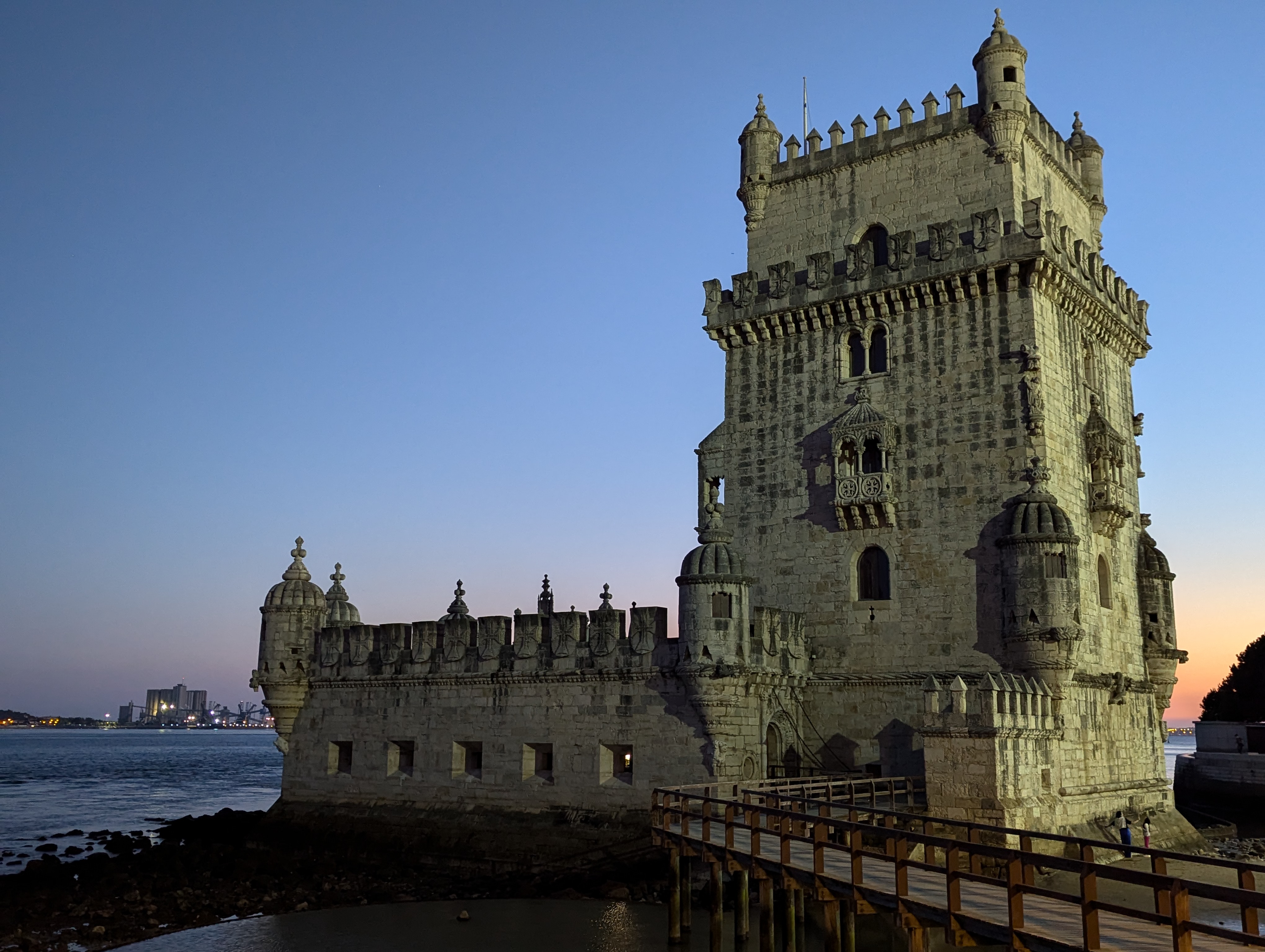}};
    \draw(2.7,2.69) node{\includegraphics[scale = 0.056,clip]
    {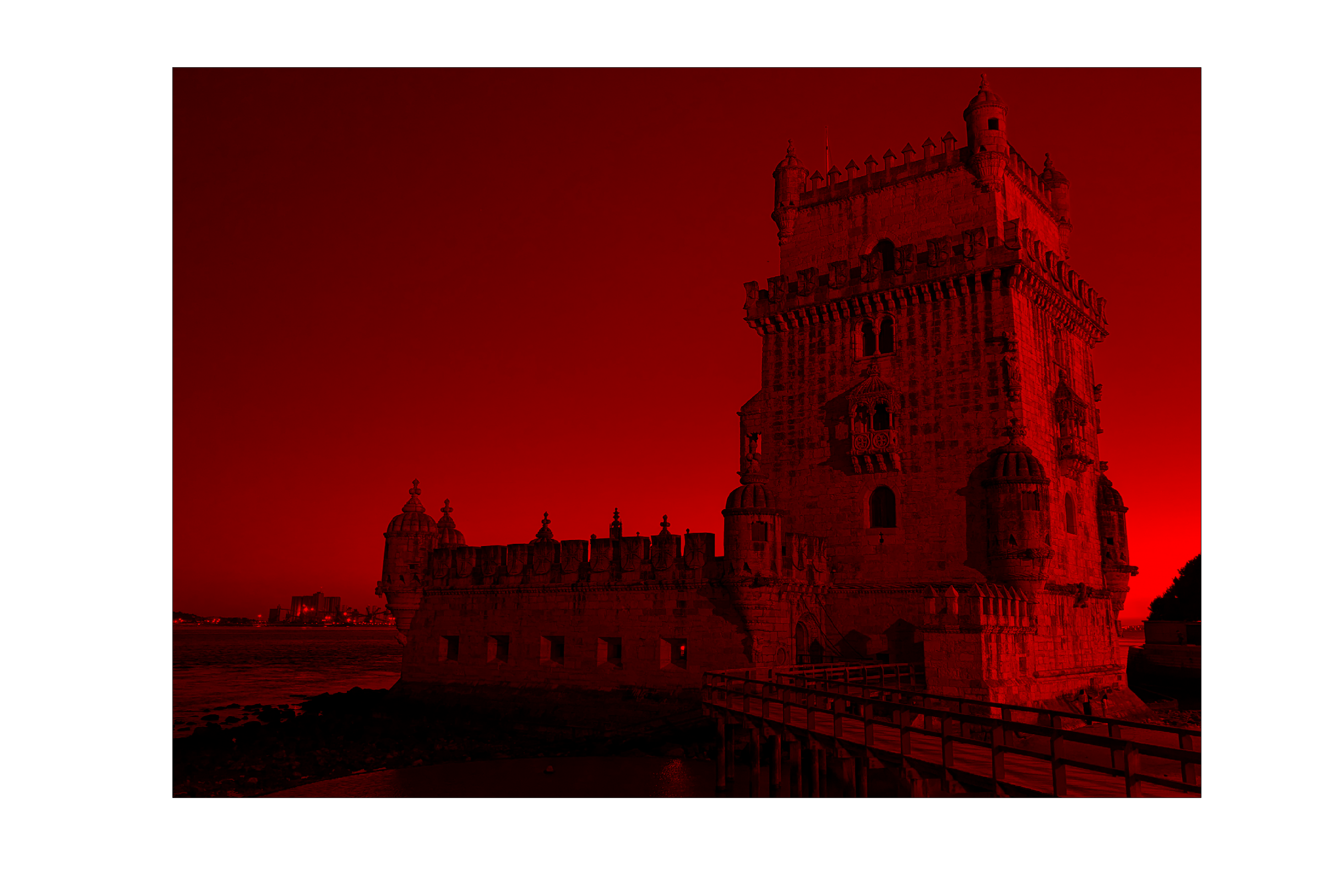}};
    \draw(2.7,0) node{\includegraphics[scale = 0.056,clip]
    {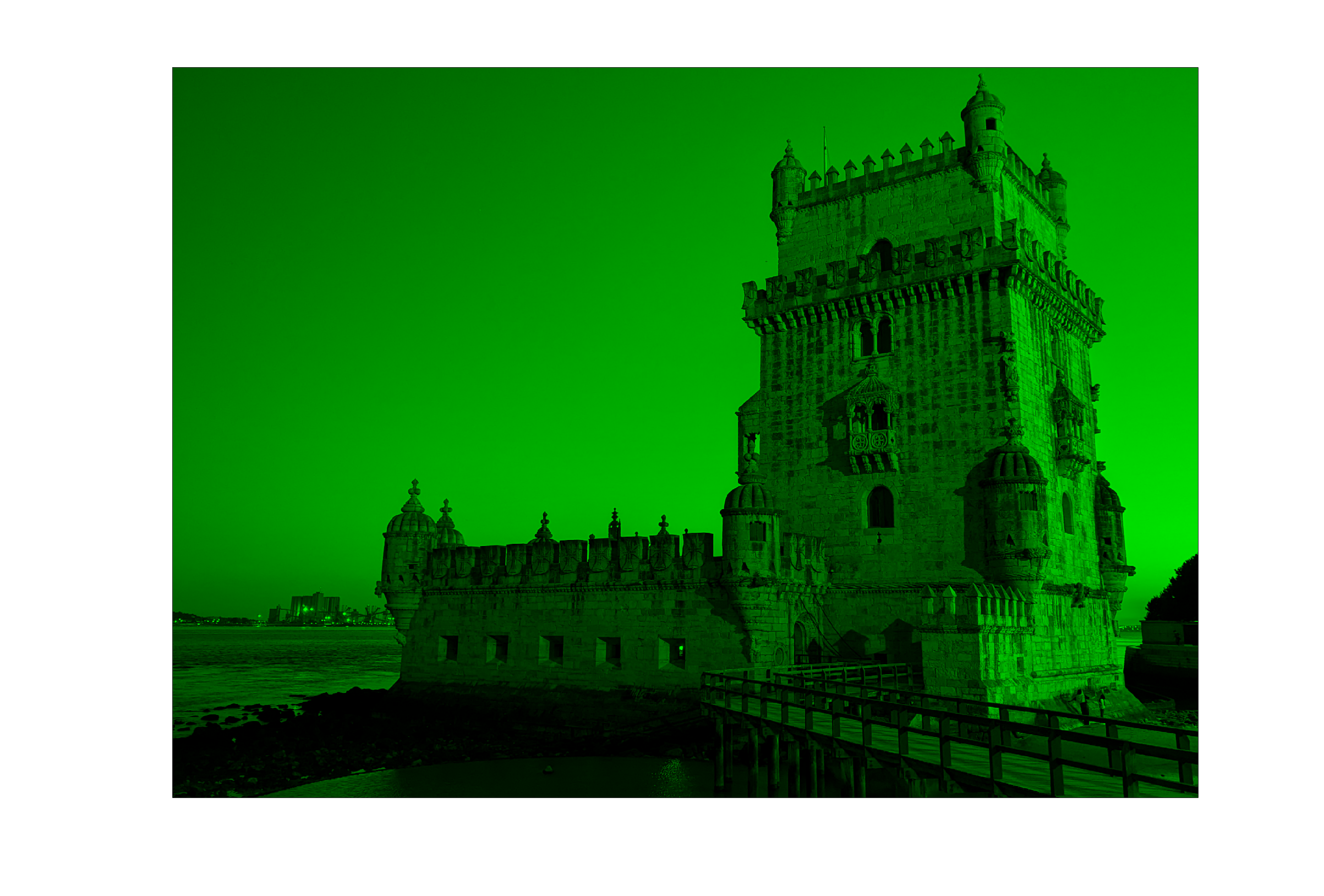}};
    \draw(2.7,-2.78) node{\includegraphics[scale = 0.056,clip]
    {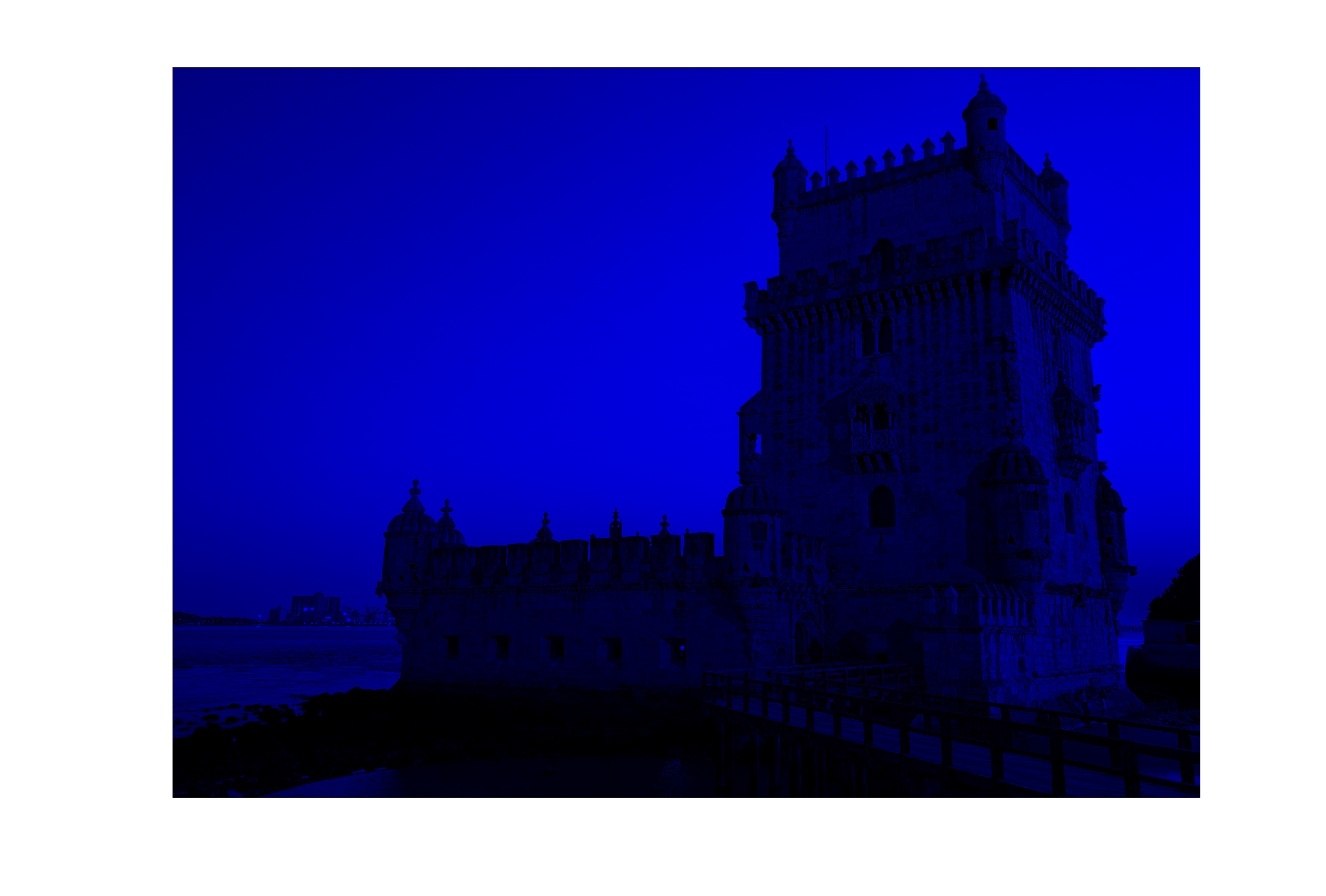}};
\end{tikzpicture}
\caption{High-resolution ($4080 \times 3072$) photograph of the Tower of 
  Bel{\'e}m in Lisbon, utilized as a multidimensional RGB test signal for 
  evaluating the $hp$-adaptive refinement strategies.
} \label{fig:Belem}
\end{figure}

\begin{figure}[htb]
\centering
\begin{tikzpicture}
\draw(-3.5,3.3) node{\includegraphics[scale = 0.16, clip,trim = 700 200 700 100]{%
  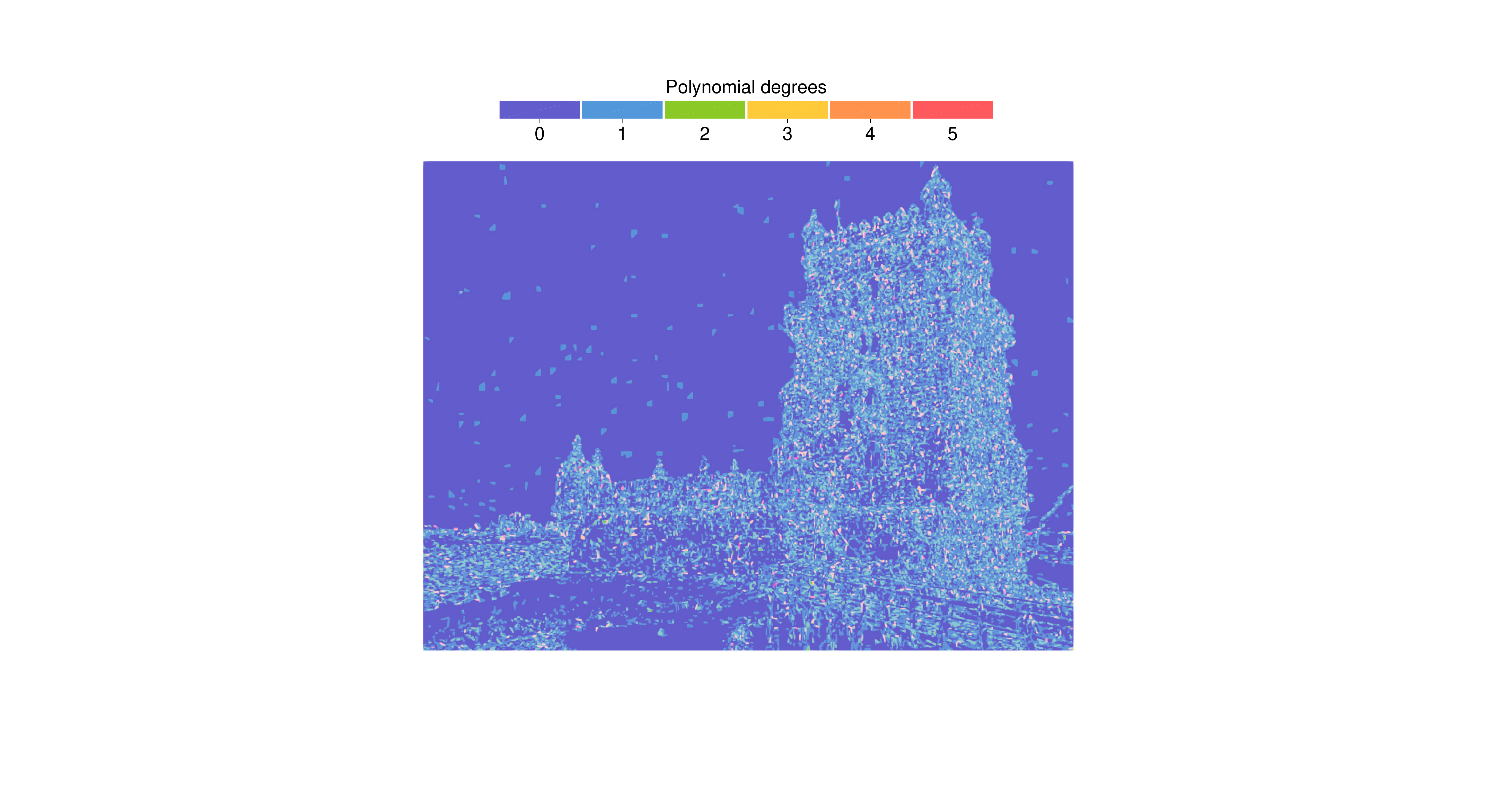}};
\draw(3.5,3.3) node{\includegraphics[scale = 0.16, clip,trim = 700 200 700 100]{%
  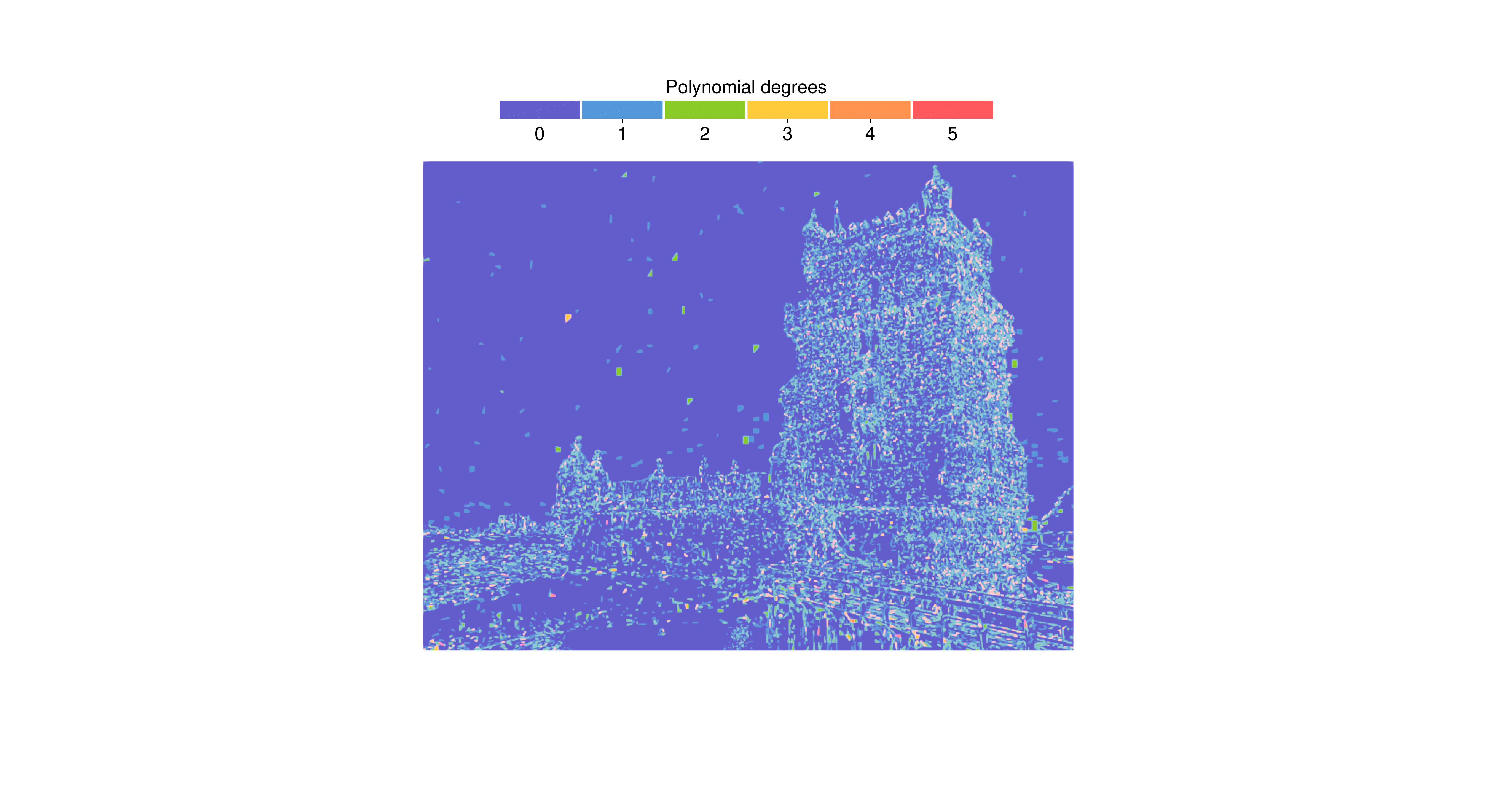}};
\draw(-3.5,-3.3) node{\includegraphics[scale = 0.16, clip,trim = 700 200 700 100]{%
  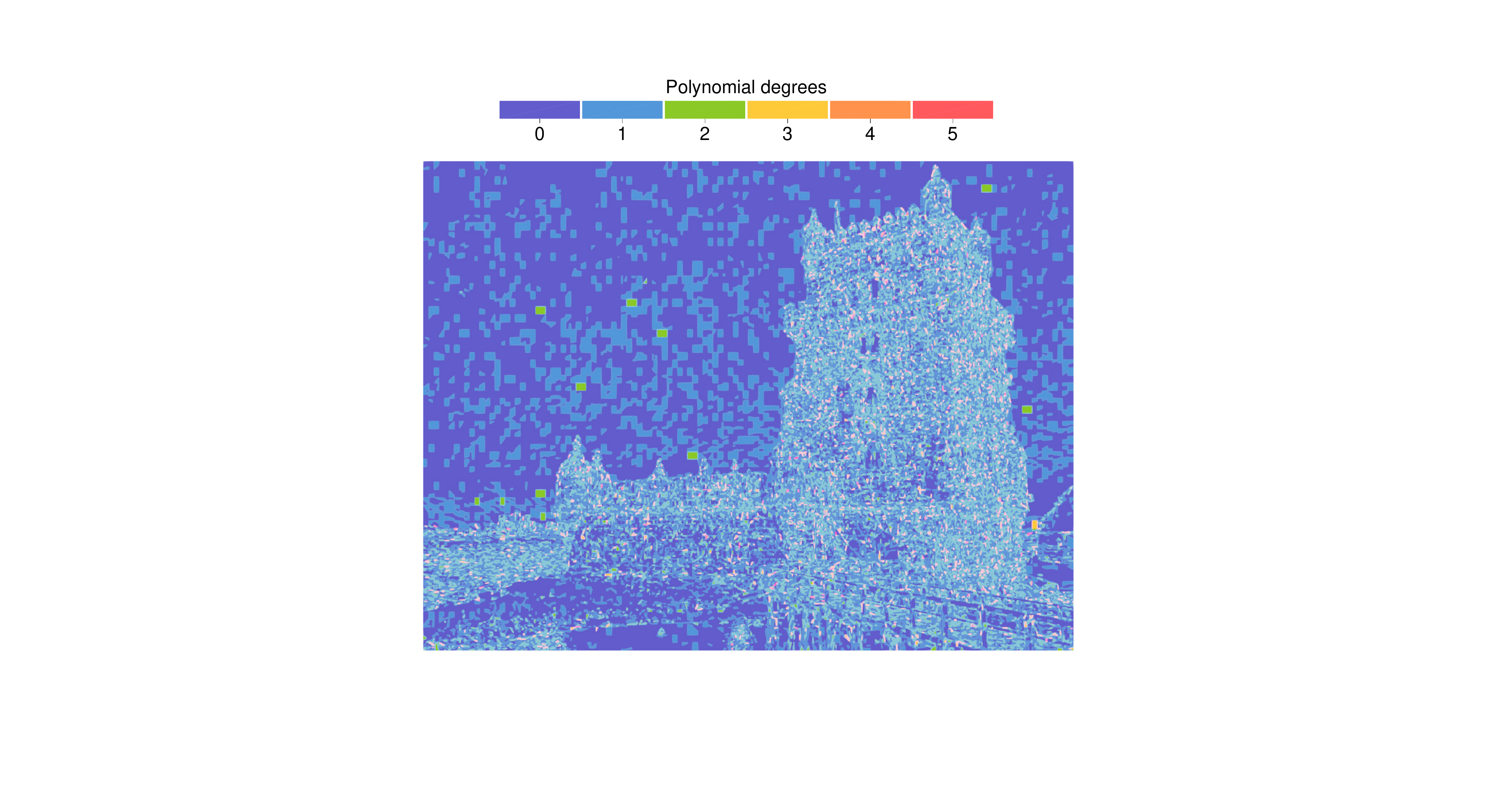}};
\draw(3.5,-3.3) node{\includegraphics[scale = 0.16, clip,trim = 700 200 700 100]{%
  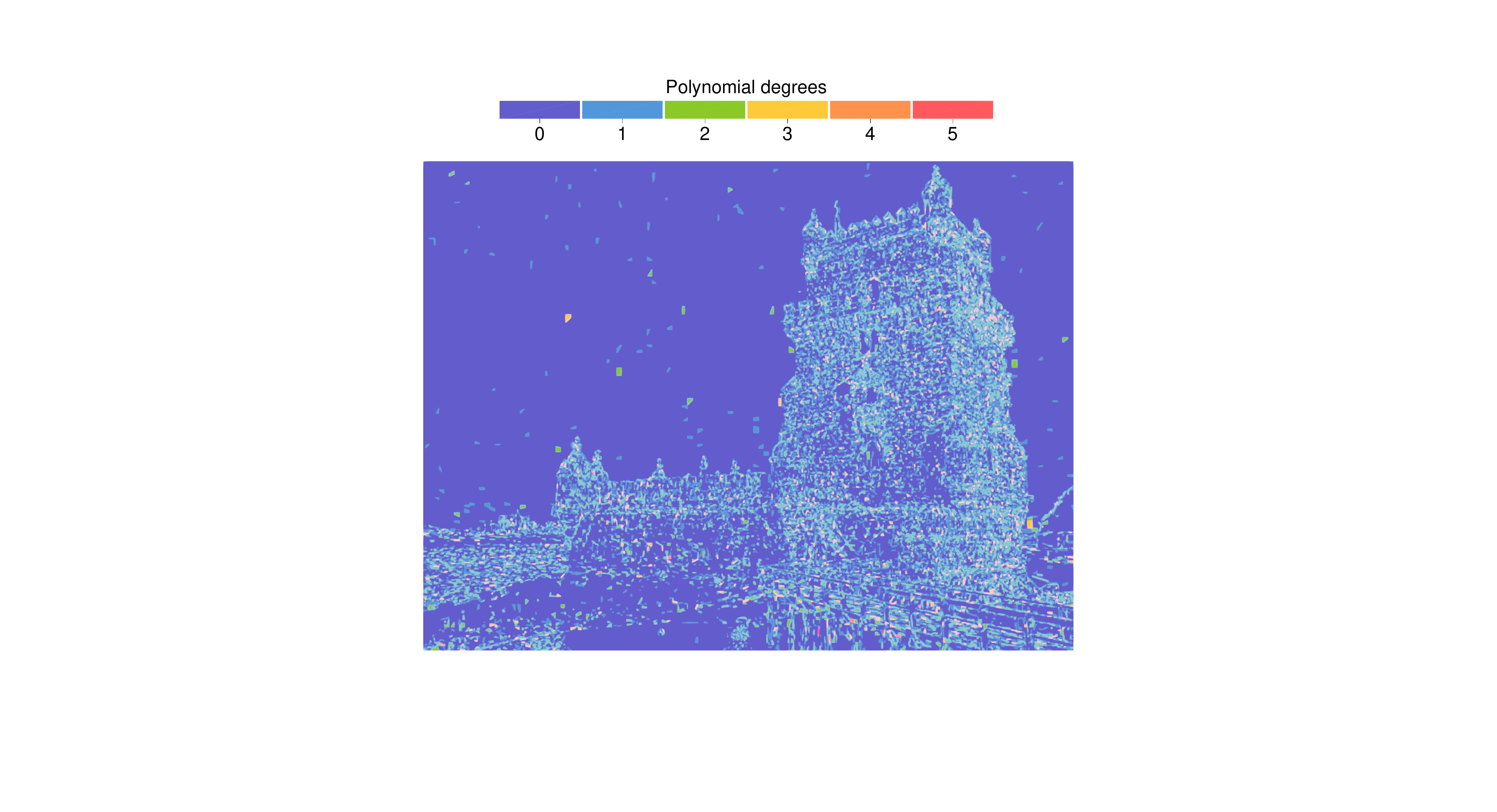}};
\end{tikzpicture}
\caption{Visualization of polynomial degree distributions across tree partition
  patches for the red channel of the Tower of Belém image, comparing the $hp$-K 
  (top-left), $hp$-ECP (top-right), $hp$-K+ECP (bottom-left) and $hp$-Binev 
  (bottom-right) approximations.
} \label{fig:poly_Belem}
\end{figure}

To validate the procedure described in Section~\ref{sec:gen_graph} for a
high-dimensional signal, we test the wedgelet forest on a synthetic
dataset. We sample $N = 3 \cdot 10^6$ uniformly distributed random
points in the domain $[1.5\pi, 4.5\pi] \times [0, 1]$ and embed them on
a manifold in $\mathbb{R}^3$. Specifically, we construct the classical
Swiss roll parametrization
\[
\begin{bmatrix}
    x(u,v) \\
    y(u,v) \\
    z(u,v)
\end{bmatrix}
=
\begin{bmatrix}
    u \cos(u) \\
    21 v \\
    u \sin(u)
\end{bmatrix}.
\]
The manifold is then lifted to $\mathbb{R}^{50}$ by taking the first
three components of each point and embedding them in the higher-
dimensional space. Subsequently, the data are transformed in
$\mathbb{R}^{50}$ by applying a random rotation, obtained via
multiplication with a random orthogonal matrix, and by adding uniform
noise of magnitude $10^{-3}$.
\begin{figure}[htb]
\centering
\begin{tikzpicture}
\draw(-3.5,0) node{\includegraphics[scale = 0.16, clip,trim = 700 150 700 0]{%
    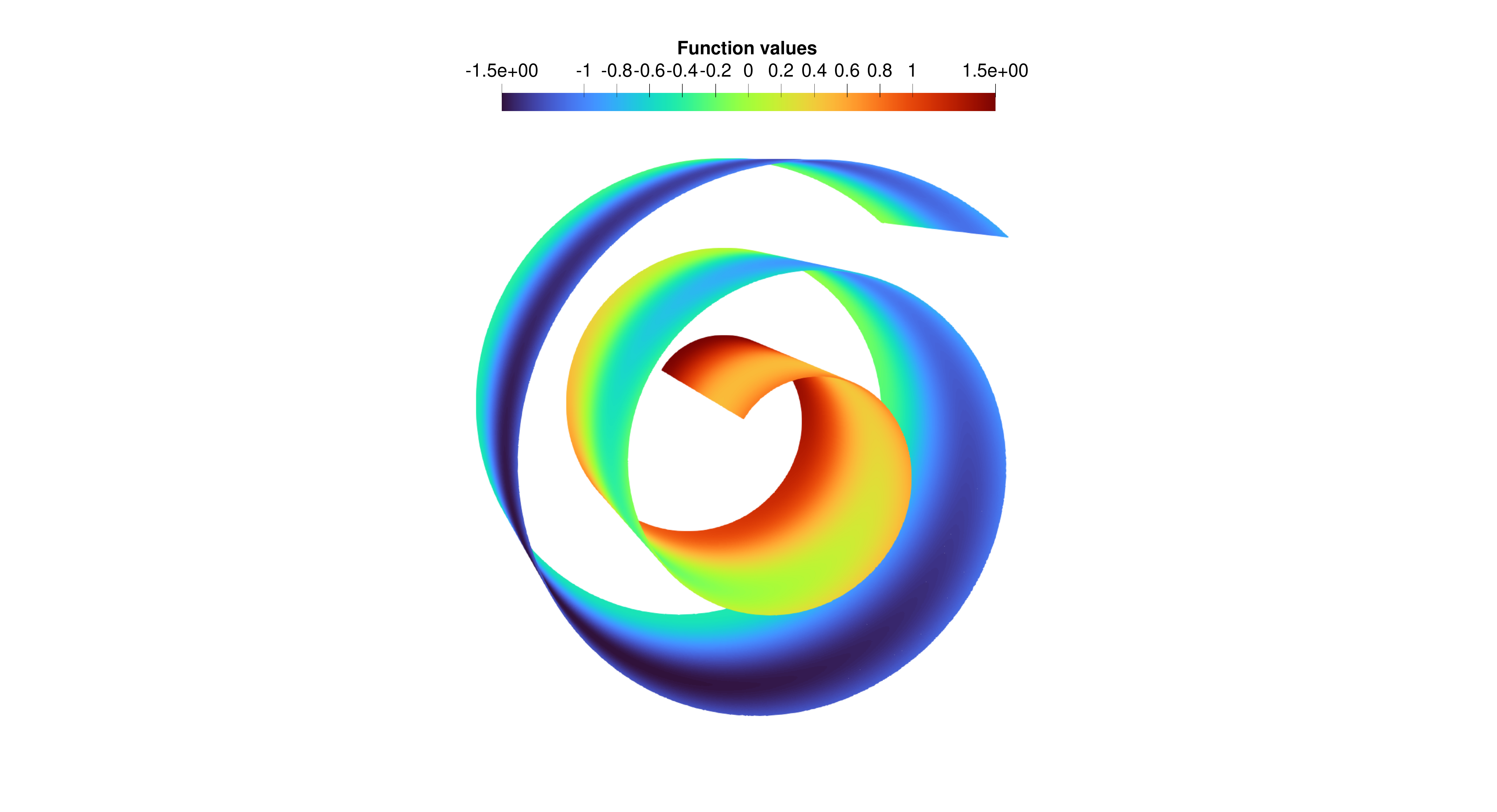}};
  \draw(3.5,0) node{\includegraphics[scale = 0.16, clip,trim = 700 150 700 0]{%
      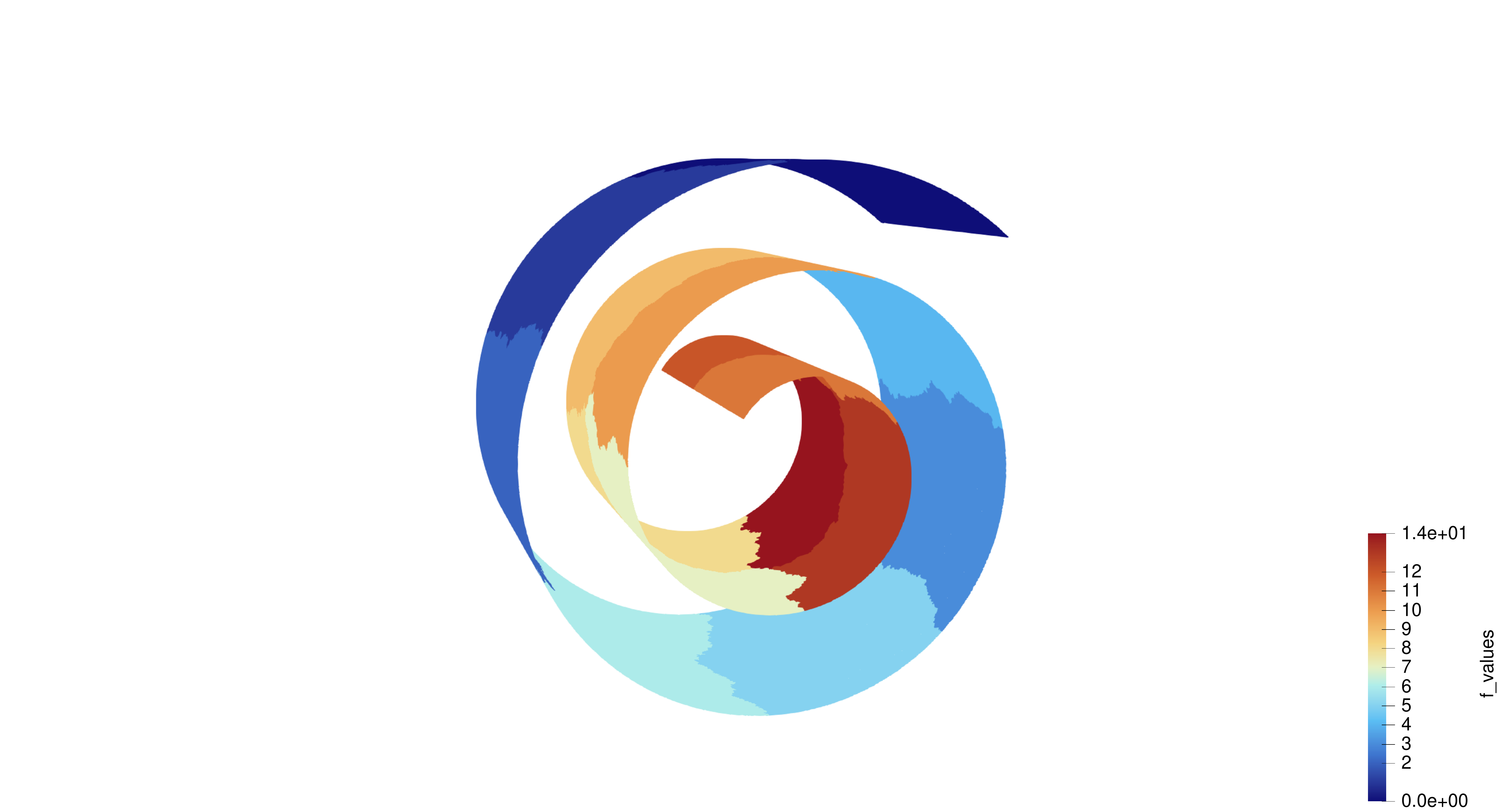}};
\end{tikzpicture}
\caption{Visualization of the dataset sampled on a 3D Swiss roll manifold,
  illustrating the signal and the partitioning used to test the wedgelet forest
  on high-dimensional data.
} \label{fig:SignalPart_SwissRoll}
\end{figure}

\begin{figure}[htb]
\centering
\begin{tikzpicture}
	\draw(-3.5,3.5) node{%
    \includegraphics[scale = 0.16, clip,trim = 700 150 700 50]{%
  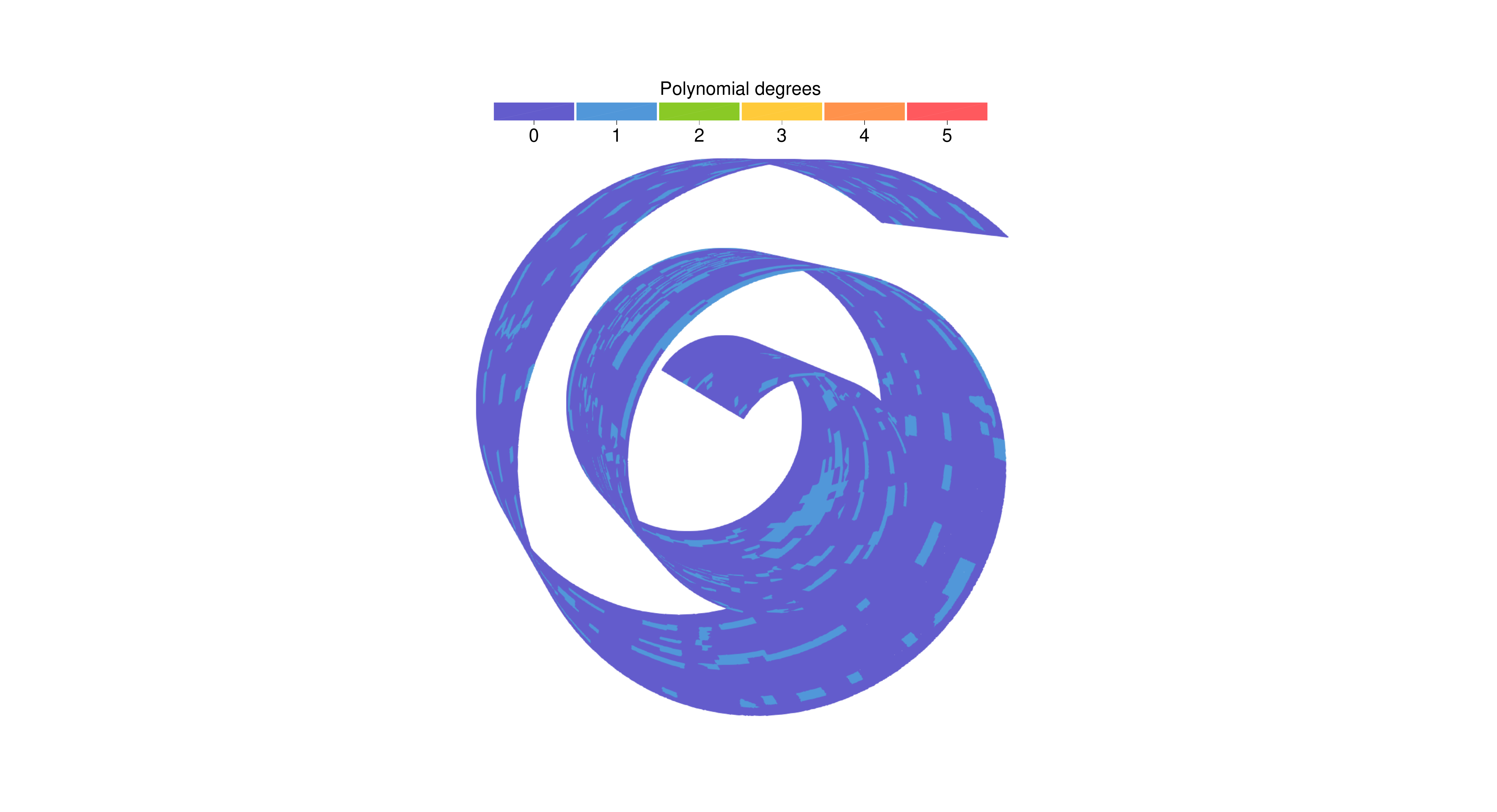}};
    \draw(3.5,3.5) node{%
      \includegraphics[scale = 0.16, clip,trim = 700 150 700 50]{%
    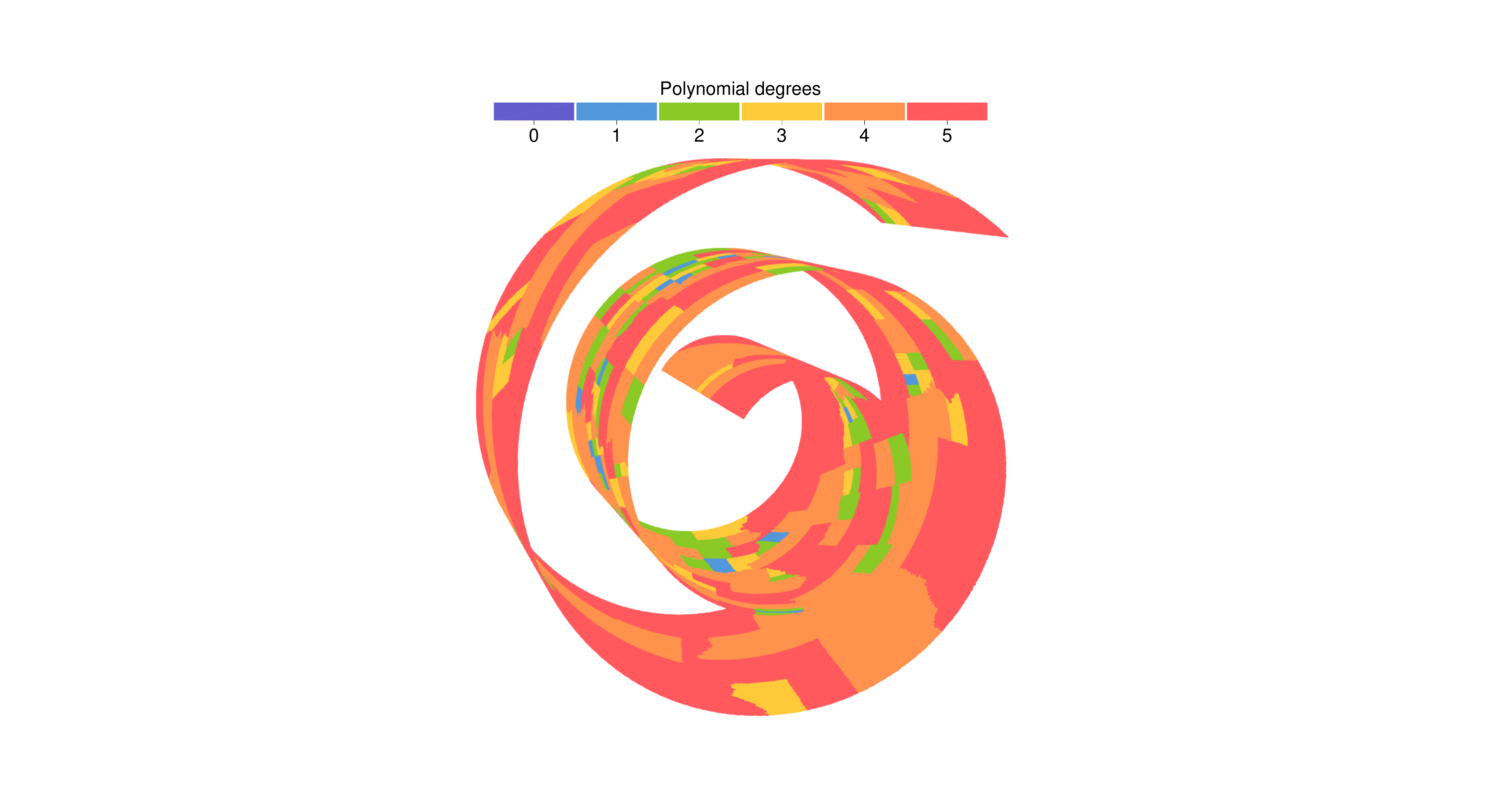}};
    \draw(-3.5,-3.5) node{%
      \includegraphics[scale = 0.16, clip,trim = 700 150 700 50]{%
    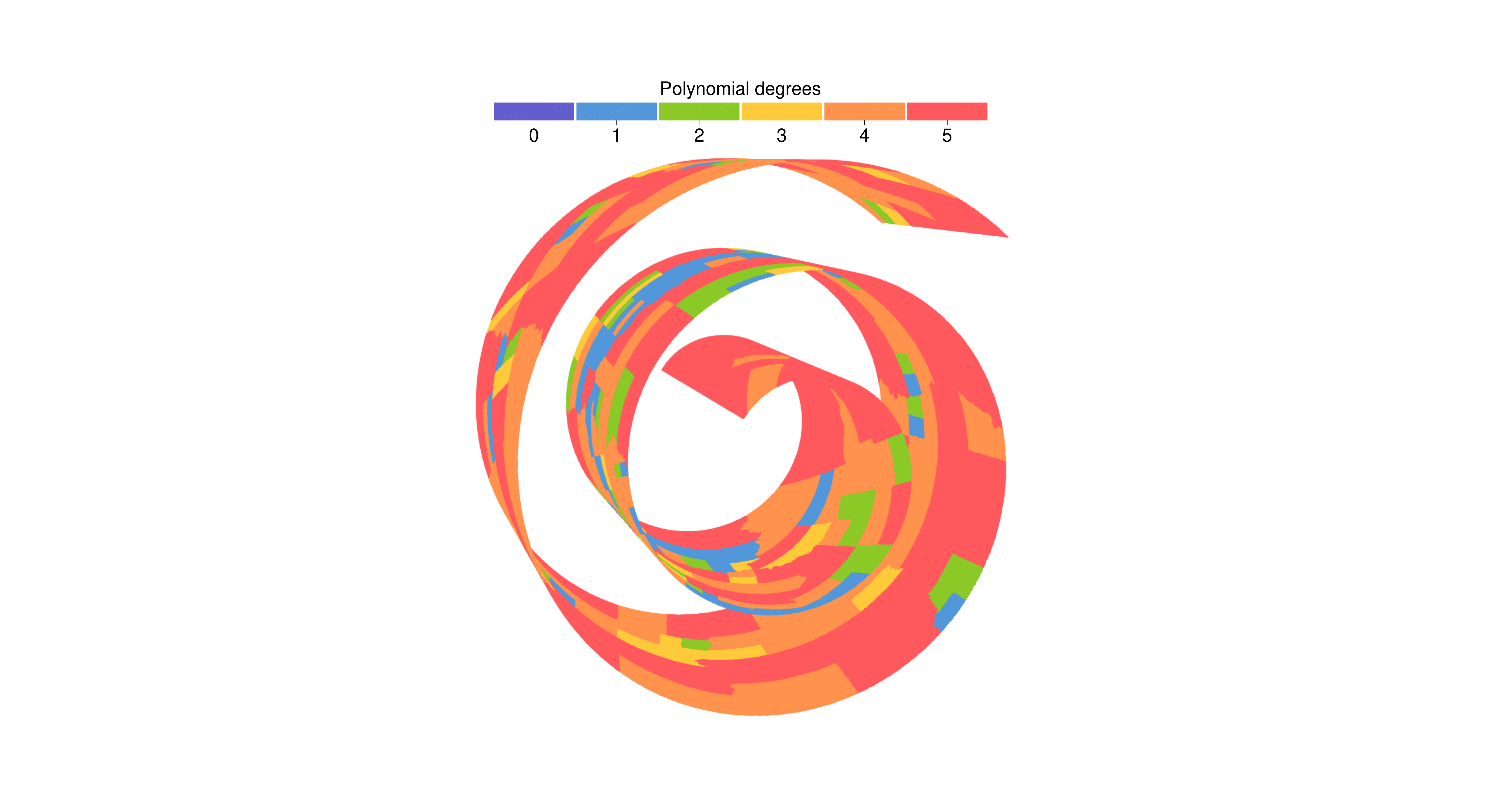}};
    \draw(3.5,-3.5) node{%
      \includegraphics[scale = 0.16, clip,trim = 700 150 700 50]{%
    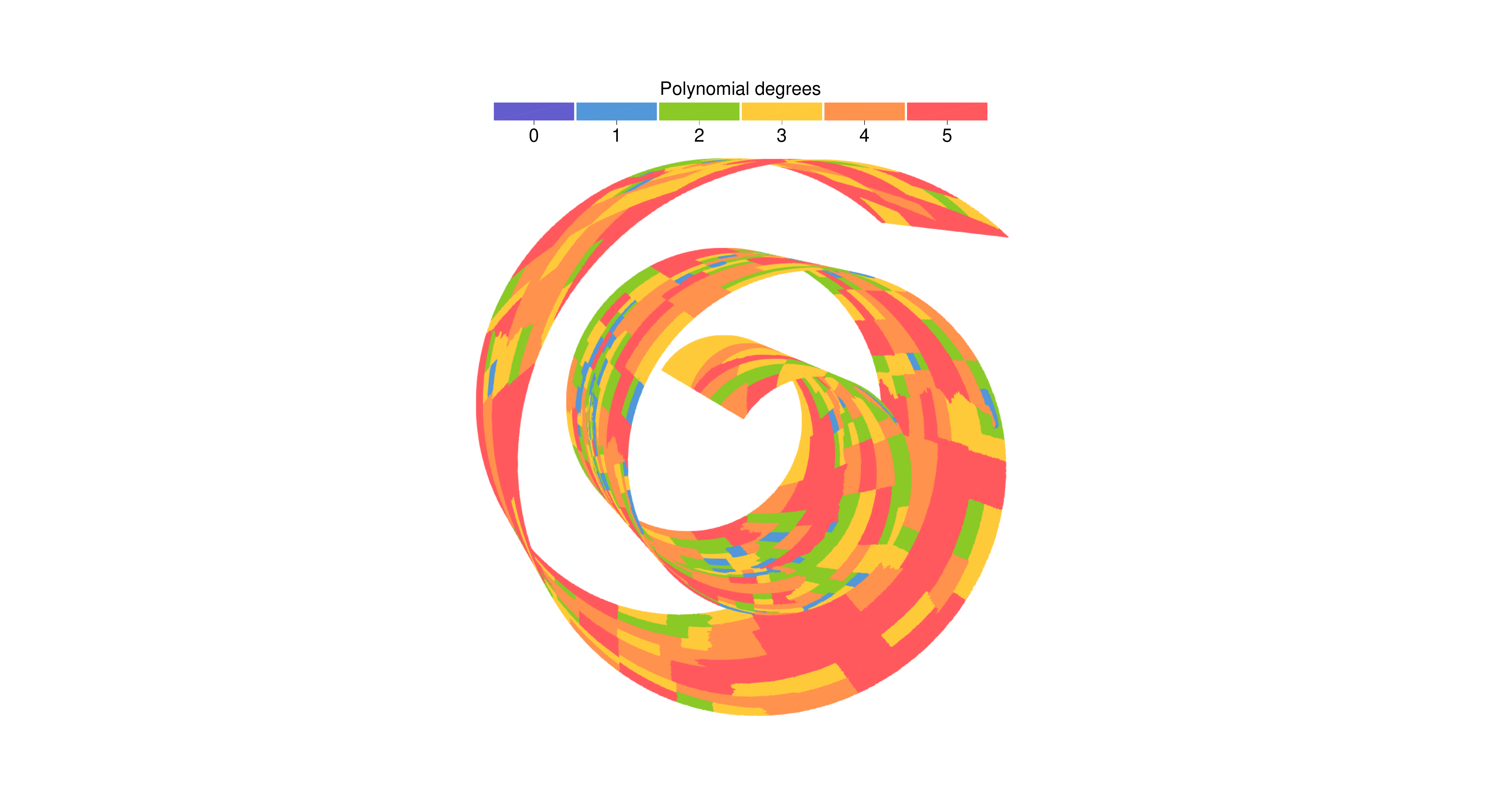}};
\end{tikzpicture}
\caption{Comparison of polynomial degrees used across local patches for the 
  lifted Swiss Roll dataset using different $hp$-adaptive refinement strategies: 
  $hp$-K (top-left), $hp$-ECP (top-right), $hp$-K+ECP (bottom-left) 
  and $hp$-Binev (bottom-right).
} \label{fig:poly_SwissRoll}
\end{figure}

On these points, we define the signal
\[
f(\mathbf{x}) = \sin\left(0.3 \sqrt{x_1^2 + x_3^2}\right)
+ 0.5 \cos(0.2\, x_2),
\]
on the Swiss roll. Following the proposed strategy, we construct the
$k$-nearest-neighbour graph with $k = 10$ and partition it into
$\rho = 15$ disjoint subgraphs $\{U_r\}_{r=1}^\rho$ using the
\texttt{METIS} algorithm. On each patch, we compute a local embedding in
$\mathbb{R}^2$ via \texttt{L-ISOMAP} with $1000$ landmarks. The signal
in $\mathbb{R}^3$ is shown in Figure~\ref{fig:SignalPart_SwissRoll} on
the left, while the resulting partition is shown on the right.

For each patch construction, we employ the same error as in the previous
global tests. The scaling is computed from the local signal,
\begin{align*}
    D_{m_r} &= \mathrm{diag}(s_{r,1}, \dots, s_{r,m}), \\
s_{r,j} &= \frac{1}{\max_{i \in U_r} |f_{i,j}|},
\end{align*}
and the patch error is normalized by the patch size $n_r = |U_r|$,
\[
\Ecal(\Tcal_r) = \frac{1}{n_r}
\left\| (\tilde{f} - f)_{|_{U_r}} D_{m_r} \right\|_F^2,
\]
where $\Tcal_r$ denotes the tree approximation on the patch $U_r$. Each
patch tree $\Tcal_r$ is refined until $\Ecal(\Tcal_r) \leq 10^{-4}$. Since 
the patches form a partition of the graph, the global error of the forest is
\begin{align*}
\Ecal(\{\Tcal_r\}_{r=1}^\rho) &= \frac{1}{N} \sum_{r=1}^\rho
\left\| (\tilde{f} - f)_{|_{U_r}} D_{m_r} \right\|_F^2 \\
&= \sum_{r=1}^\rho \frac{n_r}{N} \Ecal(\Tcal_r),    
\end{align*}
which is below the threshold $10^{-4}$ whenever each patch error meets
this tolerance.

Table~\ref{tab:complex_SwissRoll}, reports the cost of the
forest for the approximation of the signal and the total number of leaves. We obtain that in terms of storage, the greedy knapsack
criterion and ECP yield improvements of $5.5\%$ and
$57.7\%$, respectively. Their combined application results in a $58.9\%$
reduction, surpassing the $40.8\%$ gain achieved with Binev's
$hp$-near-optimal construction.

In
Figure~\ref{fig:poly_SwissRoll}, we depict the polynomial degrees used
in each partition. 
Moreover, due to
the smoothness of the signal, whenever pruning is possible, as in
Binev's strategy and the ECP algorithm, higher degree
polynomials are selected.
\begin{table}[htb]
\centering
\begin{tabular}{c|c|c|c|c|c|c}
 & $h$-max & $h$-Binev & $hp$-K & $hp$-ECP & $hp$-K+ECP & $hp$-Binev \\
\hline
 $\calC(\calT)$ &  $7\,138$ & $7\,147$  &  $7\,459$ & $5\,604$ & 
 $\bm{5\,208}$ & $7\,641$ \\
\hline
$\#\calL(\calT)$ & $7\,138$  & $7\,147$ & $6\,031$ & $\bm{431}$ & 
$663$ & $808$ \\
\hline
$\calC(\calT) + \#\calL(\calT)$ & $14\,276$  & $14\,294$ & $13\,490$ & $6\,035$ & 
$\bm{5\,871}$ & $8\,449$ 
\end{tabular}
\caption{Forest complexities and leaf cardinalities for the approximation of
the signal on the synthetic Swiss Roll dataset lifted to $\mathbb{R}^{50}$.}
\label{tab:complex_SwissRoll}
\end{table}

\section{Conclusion} \label{sec:conclude}

In this work, we have introduced an $hp$-adaptive framework for the approximation of graph signals, with a particular focus on the optimization of wedgelet-based representations. The proposed approach combines spatial $h$-refinements with novel local $p$-refinement strategies, providing a joint adaptive approximation concept exploiting simultaneously geometric and smoothness properties of the signal.

We developed two complementary refinement strategies, a greedy knapsack criterion, which selects $h$- and $p$-refinements according to their efficiency, and an error-cost pruning (ECP) algorithm, which replaces refined subtrees by higher-order polynomial approximations whenever this is more advantageous in terms of error and cost. Both strategies can be applied independently, but their combination yields in our experiments the most efficient signal representations. 

In order to have polynomial basis functions on general graphs at disposition, we applied a patch-wise local embedding strategy able to map local graph neighbourhoods into low-dimensional Euclidean spaces. This approach can be applied also for signals in high dimensional spaces and helps to reduce the computational cost associated with high-dimensional polynomial bases.

The obtained numerical experiments demonstrate that $hp$-adaptivity can substantially reduce the complexity and storage requirements of the resulting approximations compared with pure $h$-refinement strategies. In particular, combining the greedy knapsack strategy with error-cost pruning yields the most favourable option in order to reduce the number of polynomial coefficients and the size of the tree. For $hp$-adaptive wedgelet approximation, the reduction in the number of leaves leads to fewer landmarks to be stored, providing an additional benefit for signal compression.

Overall, the results indicate that the proposed $hp$-adaptive strategies provide an efficient framework for an optimized approximation and compression of signals on graphs, able to exploit signal-adaptive graph partitioning and sophisticated local polynomial approximation.

\section*{Acknowledgments}
The first and the third authors have been funded by the Swiss National Science
Foundation starting grant
``Multiresolution methods for unstructured data'' (TMSGI2\_211684). The second author acknowledges support by {GNCS-IN$\delta$AM}, as well as by the topical groups on \textquotedblleft Approximation Theory and Applications\textquotedblright\ of the Italian Mathematical Union and \textquotedblleft Math \& Brain \textquotedblright\ of the society SIMAI. 

\bibliographystyle{plain}
\bibliography{literature}

\end{document}